\documentclass{article}

\usepackage{amsmath}
\usepackage[T1]{fontenc}
\usepackage{amssymb}
\usepackage{amsthm}
\usepackage{float}
\usepackage{tikz}
\usetikzlibrary{automata}
\usetikzlibrary{shadows}
\usetikzlibrary{arrows}
\usetikzlibrary{shapes}
\usetikzlibrary{decorations.pathmorphing}
\usetikzlibrary{fit}
\usetikzlibrary{trees}
\usetikzlibrary{intersections}

\tikzstyle{every picture}=[
  >=stealth', shorten >=1pt, node distance=1.44cm,auto,bend angle=45,initial text=,
  every state/.style={inner sep=0.75mm, minimum size=1mm},font=\scriptsize,
]

\newcommand{\ocap}{\otimes}
\newcommand{\ocup}{\oplus}
\newcommand{\oneg}{\ominus}
\newcommand{\bigwedgee}{{\bigwedge_e}}

\newcommand{\Intervalle}[2]{\{#1,\ldots,#2\}}

\newtheorem{definition}{Definition}
\newtheorem{proposition}{Proposition}
\newtheorem{corollary}{Corollary}
\newtheorem{theorem}{Theorem}
\newtheorem{example}{Example}

\begin{document}

  \title{A General Framework for the Derivation of Regular Expressions}

  \author{
    Pascal Caron, Jean-Marc Champarnaud and Ludovic Mignot\\
    \{pascal.caron,jean-marc.champarnaud,ludovic.mignot\}@univ-rouen.fr\\
    LITIS, Universit\'e de Rouen,\\
     76801 Saint-\'Etienne du Rouvray Cedex, France
  }

  \maketitle

   \begin{abstract}
The aim of this paper is to design a theoretical 
framework
that allows us to perform the computation of regular expression derivatives through a space of generic structures.
Thanks to this formalism, the main properties of regular expression derivation, such as 
the finiteness of the set of derivatives,
need only be stated and proved one time, at the top level. 
Moreover, it is shown how to
construct an alternating automaton
associated with the derivation of a regular expression in this general 
framework.
Finally, Brzozowski's derivation and Antimirov's derivation turn out to be a particular case of this general scheme
and it is shown how to construct a DFA, a NFA and an AFA for both of these derivations.
  \end{abstract}

\section{Introduction}

The (left) quotient of a language $L$ over an alphabet $\Sigma$ with respect to a word $w$ in $\Sigma^*$
is the language obtained by stripping the leading $w$ from the words in $L$ that are prefixed by $w$. 
The quotient operation plays a fundamental role in language theory and is especially involved in two main issues.
First, checking whether a word $w$ belongs to a language $L$
turns out to be equivalent to checking whether the empty word belongs to the quotient of $L$ w.r.t. $w$. 
Secondly it was proved by Myhill~\cite{Myh57} and Nerode~\cite{Ner58} that a language is regular
if and only if the set of its quotients w.r.t. all the words in $\Sigma^*$ is finite.
In the case of a regular language $L$, 
the different quotients are the states of the so-called quotient automaton of $L$
that is isomorphic to its minimal deterministic automaton.

Since the equality of two languages amounts to 
the isomorphism of their minimal deterministic automata, 
the construction of the quotient automaton via the computation of the left quotients is intractable.
The seminal work of Brzozowski~\cite{Brz64}
that introduced the notion of a word derivative of a regular expression
and the construction of the associated DFA,
gave rise to a long series of studies~(see for example~\cite{BS86,Ant96,CZ01c,IY03,CCM11b}) 
that are all based on 
the simulation of the computation of a language 
quotient
by the one of an expression derivative.
In all these research works, the first rule is that if the expression $E$ denotes the language $L$,
then the derivative of $E$ w.r.t. $w$, for any $w$, denotes the quotient of $L$ w.r.t. the word $w$.
Thus, checking whether a word $w$ belongs to the language denoted by the expression $E$  is equivalent to 
checking if the empty word belongs to the language denoted by the derivative of $E$ w.r.t. $w$.
The second rule is that as far as the set $D$ of all the derivatives of $E$ is finite,
a finite automaton recognizing the language denoted by $E$ can be constructed,
admitting $D$ as a set of states.

Let us notice 
that Brzozowski derivatives~\cite{Brz64} handle 
unrestricted regular expressions and provide a deterministic automaton;
Antimirov derivatives~\cite{Ant96} only address simple regular expressions 
and provide both a deterministic automaton and a non-deterministic one;
Antimirov derivatives have been recently extended to 
regular expressions~\cite{CCM11b} and this extension
provides a deterministic automaton, a non-deterministic one
and, as shown in this paper, an alternating automaton.
Berry and Sethi continuations~\cite{BS86} are based on the linearization of the (simple) input expression
and allow the construction of its Glushkov (non-deterministic) automaton.
Champarnaud and Ziadi c-continuations~\cite{CZ01c} and Ilie and Yu derivatives~\cite{IY03}
allow both the construction of the Glushkov automaton and of the Antimirov non-deterministic automaton.
Let us mention that derivation has been extended to expressions with multiplicity~\cite{LS01,COZ09}. 
 
As mentioned by Antimirov~\cite{Ant96},
derivatives of regular expressions have proved to be a productive concept to investigate theoretical topics such as
the algebra of regular expressions~\cite{Con71} or
of $K$-regular expressions~\cite{Kro92},
the systems of language equations~\cite{BL80},
the equivalence of simple regular expressions~\cite{Gin67} or of regular expressions~\cite{AM95}.
More recently, Brzozowski introduced a new approach for finding upper bounds
for the state complexity of regular languages,
based on the counting of their quotients (or of their derivatives)~\cite{Brz10}.

Moreover, derivatives provide a useful tool to implement regular matching algorithms:
Brzozowski's DFA and Antimirov's NFA turn out to be 
competitive matching automata~\cite{SL07},
compared for instance with Thompson's $\varepsilon$-automaton~\cite{Tho68}.
The derivative-based techniques are well-suited to functional languages,
that are characterized by a good support for symbolic term manipulation.
As an example, two derivative-based scanner generators have been recently developed,
one for PLT Scheme and one for Standard ML, as reported in~\cite{ORT09}.
Similarly, Brzozowski's derivatives are used in the implementation of the XML schema language RELAX NG~\cite{Cla02}.
Finally, let us notice that derivatives can be extended to context-free grammars, seen as recursive regular expressions,
yielding a system for parsing context-free grammars~\cite{MDS11}.

The aim of this paper is to design a general 
framework
where the computation of the set of derivatives of a regular expression, called derivation, is performed over a space of generic structures.
Of course Brzozowski's derivation and Antimirov's one appear as particular cases of this general scheme.
A first benefit of this formalism is that the properties inherent to the mechanism of derivation, such as 
the equality between the language denoted by a derivative and the corresponding quotient,
the finiteness of the set of derivatives
and the way for constructing the associated automata,
need only be stated and proved one time, at the top level. 
A second benefit is that the general 
framework
allows us to design the construction of an AFA from the set of derivatives.
As a consequence, we show how to construct a DFA, a NFA and an AFA for any finite derivation, including both Brzozowski's one and Antimirov's one.

The next section is a preliminary section; it gathers classical notions concerning 
regular languages, regular expressions and finite automata, as well as boolean formulas and alternating automata.
The notion of a regular expression derivation via a support is defined in Section~\ref{sec deriv via},
and the properties of the corresponding derivatives are investigated. 
Section~\ref{sec from der to aut} is devoted to the construction of the alternating automaton
associated with the derivation of a regular expression via a support.
This construction is illustrated in Section~\ref{sec deriv via set clausal form},
where the support is based on the set of clausal forms over the alphabet of the regular expressions.

\section{Preliminaries}\label{sec prelim}

  Let $\mathbb{B}=\{0,1\}$. A \emph{boolean formula} $\phi$ over a set $X$ is inductively defined by $\phi=x$ where $x\in X$, or $\phi=\mathrm{f}_{\mathbb{B}}(\phi_1,$ $\ldots, \phi_k)$, where $\mathrm{f}_{\mathbb{B}}$ is the operator associated to the $k$-ary boolean function $\mathrm{f}$ from $\mathbb{B}^k$ to $\mathbb{B}$ and $\phi_1,\ldots,\phi_k$ are boolean formulas over $X$. The set of the boolean formulas over $X$ is denoted by $\mathrm{BoolForm}(X)$.
    Let $\mathrm{v}$ be a function from $X$ to $\mathbb{B}$. The \emph{evaluation of} $\phi$ \emph{with respect to} $\mathrm{v}$ is the boolean $\mathrm{eval}_{\mathrm{v}}(\phi)$ inductively defined by: $\mathrm{eval}_{\mathrm{v}}(x)=\mathrm{v}(x)$, $\mathrm{eval}_{\mathrm{v}}(\mathrm{f}_{\mathbb{B}}(\phi_1,\ldots,\phi_k))=\mathrm{f}(\mathrm{eval}_{\mathrm{v}}(\phi_1),\ldots,\mathrm{eval}_{\mathrm{v}}(\phi_k))$. The set $\mathrm{Atom}(\phi)$ is the subset of $X$ inductively defined by $\mathrm{Atom}(x)=\{x\}$ and $\mathrm{Atom}(\mathrm{f}_{\mathbb{B}}(\phi_1,\ldots,\phi_k))=\bigcup_{1\leq j\leq k}\mathrm{Atom}(\phi_j)$.
    
    Let $\Sigma$ be an alphabet, $w$ be a word in $\Sigma^*$ and $L$ be a language over $\Sigma$. Deciding whether $w$ belongs to $L$ is called the \emph{membership problem for the language} $L$. We denote by $\mathrm{r}_w(L)$ the boolean equal to $1$ if $w\in L$, $0$ otherwise. Let $L_1,\ldots,L_k$ be $k$ languages. We denote by $\cdot$, $^*$ and $\mathrm{f}_L$ for any $k$-ary boolean function $\mathrm{f}$ the operators defined as follows:
     $L_1\cdot L_2=\{w_1\cdot w_2 \in\Sigma^{*}\mid r_{w_1}(L_1) \wedge \mathrm{r}_{w_2}(L_2)=1\}$,
     $L_1^{*}=\{\varepsilon\}\cup \{w_1\cdots w_n\in\Sigma^{*}\mid n\in\mathbb{N}\wedge \forall k\in\Intervalle{1}{n}, \mathrm{r}_{w_k}(L_1)=1\}$,     
      $\mathrm{f}_L(L_1,\ldots,L_k)=\{w\in \Sigma^{*}\mid \mathrm{f}(\mathrm{r}_w(L_1),\ldots,\mathrm{r}_w(L_k))=1\}$. The \emph{quotient of} a regular language $L$ with respect to a word $w$, that is defined as the set $w^{-1}(L)=\{w'\in\Sigma^* \mid \mathrm{r}_{ww'}(L)=1\}$ can be inductively computed as follows: 
$\varepsilon^{-1}(L)=L$,  and for $a\in \Sigma$,

	    $a^{-1}(L)=
	      \left\{
	        \begin{array}{l@{\ }l}
	          a^{-1}(L_1)\cdot  L_2 \cup a^{-1}(L_2) & \text{ if } L=L_1 \cdot  L_2\wedge \varepsilon\in L_1,\\
	          a^{-1}(L_1)\cdot  L_2 & \text{ if } L=L_1\cdot  L_2\wedge \varepsilon\notin L_1,\\
	          a^{-1}(L_1)\cdot  L_1^{*} & \text{ if } L=L_1^{*},\\
	          \mathrm{f}_L(a^{-1}(L_1),\ldots,a^{-1}(L_n)) & \text{ if } L=\mathrm{f}_L(L_1,\ldots,L_n),\\
	          \{\varepsilon\} & \text{ if } L=\{a\},\\
	          \emptyset & \text{ otherwise.}\\
	        \end{array}
	      \right.$

$(a\cdot w')^{-1}(L)=w'^{-1}(a^{-1}(L))$ for $w'\in \Sigma^+$.

  An \emph{Alternating Automaton} (\textbf{AA}) is a $5$-tuple $A=(\Sigma,Q,I,F,\delta)$ where $\Sigma$ is an \emph{alphabet}, $Q$ is a set of states, $I$ is a boolean formula over $Q$, $F$ is a function from $Q$ to $\mathbb{B}$ and $\delta$ is a function from $Q\times\Sigma$ to $\mathrm{BoolForm}(Q)$. The function $\delta$ is extended from $\mathrm{BoolForm}(Q)\times \Sigma^*$ to $\mathrm{BoolForm}(Q)$ as follows: $\delta(\phi,\varepsilon)=\phi$, $\delta(\phi,aw)=\delta(\delta(\phi,a),w)$, $\delta(\mathrm{f}_{\mathbb{B}}(\phi_1,\ldots,\phi_k),a)=\mathrm{f}_{\mathbb{B}}(\delta(\phi_1,a),\ldots,\delta(\phi_k,a))$ where $a$ is any symbol in $\Sigma$, $w$ is any word in $\Sigma^*$ and $\phi_1,\ldots,\phi_k$ are any $k$ boolean formulas over $Q$.
  The \emph{accessible part of} $A$ is the alternating automaton $(\Sigma,Q',I,F',\delta')$ defined by: 
  $Q'=\{q\in Q\mid\exists w\in\Sigma^{*}, q\in\mathrm{Atom}(\delta(I,w))\}$;
  $\forall q\in Q'$, $F'(q)=F(q)$; $\forall a\in\Sigma, \forall q\in Q', \delta'(q,a)=\delta(q,a)$.  
  The \emph{language recognized} by the alternating automaton $A$ is the subset $L(A)$ of $\Sigma^*$ defined by $L(A)=\{w\in\Sigma^*\mid\mathrm{eval}_F(\delta(I,w))=1\}$.
  Whenever $Q$ is a finite set, $A$ is said to be an \emph{Alternating Finite state Automaton} (\textbf{AFA}).   
    
  A \emph{regular expression} $E$ over an alphabet $\Sigma$ is inductively defined by: $E=a$, $E=1$, $E=0$, $E=E_1\cdot E_2$, $E=E_1^*$ or $E=\mathrm{f}_{\mathrm{e}}(E_1,\ldots,E_n)$, where $\forall k\in \mathbb{N}$, $E_k$ is a regular expression, $a\in\Sigma$ and $\mathrm{f}_{\mathrm{e}}$ is the operator associated to the $k$-ary boolean function $\mathrm{f}$, \emph{e.g.} $+$ is the operator associated to $\vee$.
  In the following, we assume that the regular expression operators as well as the boolean formula operators have no specific algebraic properties, unlike the boolean functions. For instance, the operator $+$ is not associative, not commutative, nor idempotent. 
 A regular expression is said to be \emph{simple} if the only boolean operator used is the sum.  
 
  The set of the regular expressions over an alphabet $\Sigma$ is denoted by $\mathrm{Exp}(\Sigma)$.
  The \emph{language denoted by} $E$ is the subset $L(E)$ of $\Sigma^{*}$ inductively defined as follows: $L(E\cdot F)=L(E)\cdot L(F)$, $L(E^*)=(L(E))^{*}$, $L(\mathrm{f}_{\mathrm{e}}(E_1,\ldots,E_n))=\mathrm{f}_L(L(E_1),\ldots,L(E_n))$, $L(a)=\{a\}$, $L(0)=\emptyset$, and $L(1)=\{\varepsilon\}$ with $E$ and $F$ any two regular expressions, $a$ any symbol in $\Sigma$ and $\mathrm{f}_L$ the operator associated with $\mathrm{f}$ (\emph{e.g.} $\cup$ is associated to $\vee$). Whenever two expressions $E_1$ and $E_2$ denote the same language, $E_1$ and $E_2$ are said to be \emph{equivalent} (denoted by $E_1\sim E_2$).
   In the following, we denote by $\mathrm{r}_w(E)$ the boolean $\mathrm{r}_w(L(E))$.	  
 Notice that the boolean $r_\varepsilon(E)$ is straightforwardly computed as follows:
	  
	  \centerline{
	    $\forall \alpha\in\Sigma\cup\{1,0\}$, $r_\varepsilon(\alpha)=
	      \left\{
	        \begin{array}{l@{\ }l}
	          1 & \text{ if } \alpha=1,\\
	          0 & \text{ otherwise,}\\
	        \end{array}
	      \right.$
	  }
	  
	  \centerline{
	    $r_\varepsilon(\mathrm{f}_{\mathrm{e}}(E_1,\ldots,E_k))=\mathrm{f}(r_\varepsilon(E_1),\ldots,r_\varepsilon(E_k))$,
	  }
	  
	  \centerline{
	    $r_\varepsilon(E_1\cdot E_2)= r_{\varepsilon}(E_1)\wedge r_{\varepsilon}(E_2),$
	  }
	  
	  \centerline{
	    $r_\varepsilon(E_1^*)=1$.
	  }
	  
Given a regular expression $E$ over an alphabet $\Sigma$ and a symbol $a$ in $\Sigma$, the \emph{Brzozowski derivative} of $E$ w.r.t. $a$ is the expression $\frac{d}{d_a}(E)$ inductively defined by:

\centerline{ $\frac{d}{d_a}(a)=1$, $\frac{d}{d_a}(b)=\frac{d}{d_a}(1)=\frac{d}{d_a}(0)=0$,}

\centerline{ $\frac{d}{d_a}(\mathrm{f}_e(E_1,\ldots,E_k))=\mathrm{f}_e(\frac{d}{d_a}(E_1),\ldots,\frac{d}{d_a}(E_k))$, $\frac{d}{d_a}(E_1^*)=\frac{d}{d_a}(E_1)\cdot E_1^*$,}

\centerline{
  $\frac{d}{d_a}(E_1\cdot E_2)=
    \left\{
      \begin{array}{l@{\ }l}
        \frac{d}{d_a}(E_1)\cdot E_2 + \frac{d}{d_a}(E_2) & \text{ if } \mathrm{r}_\varepsilon(E_1)=1,\\
        \frac{d}{d_a}(E_1)\cdot E_2 & \text{ otherwise,}\\
      \end{array}
    \right.
  $}
  
  where $b$ is any symbol in $\Sigma\setminus\{a\}$, $E_1,\ldots,E_k$ are any $k$ regular expressions over $\Sigma$, and $\mathrm{f}_e$ is the operator associated with the $k$-ary boolean function $\mathrm{f}$. Notice that Brzozowski defines the \emph{dissimilar derivative} of $E$ as the expression $\frac{d'}{d'_a}(E)$ inductively computed by substituting the operator $+_{ACI}$ to the $+$ operator in the derivative formulas, where $+_{ACI}$ is the associative, commutative and idempotent version of $+$.
  
  Given a simple regular expression $E$ over an alphabet $\Sigma$ and a symbol $a$ in $\Sigma$, the \emph{Antimirov partial derivative} of $E$ w.r.t. $a$ is the set of expression $\frac{\partial}{\partial_a}(E)$ inductively defined by:

\centerline{ $\frac{\partial}{\partial_a}=\{1\}$, $\frac{\partial}{\partial_a}(b)=\frac{\partial}{\partial_a}(1)=\frac{\partial}{\partial_a}(0)=\emptyset$,}

\centerline{ $\frac{\partial}{\partial_a}(E_1+E_2)=\frac{\partial}{\partial_a}(E_1)\cup\frac{\partial}{\partial_a}(E_2)$, $\frac{\partial}{\partial_a}(E_1^*)=\frac{\partial}{\partial_a}(E_1)\cdot E_1^*$,}

\centerline{
  $\frac{\partial}{\partial_a}(E_1\cdot E_2)=
    \left\{
      \begin{array}{l@{\ }l}
        \frac{\partial}{\partial_a}(E_1)\cdot E_2 \cup \frac{\partial}{\partial_a}(E_2) & \text{ if } \mathrm{r}_\varepsilon(E_1)=1,\\
        \frac{\partial}{\partial_a}(E_1)\cdot E_2 & \text{ otherwise,}\\
      \end{array}
    \right.
  $}
  
  where $b$ is any symbol in $\Sigma\setminus\{a\}$, $E_1,\ldots,E_k$ are any $k$ regular expressions over $\Sigma$, and for any set $\mathcal{E}$ of regular expression, for any regular expression $F$, $\mathcal{E}\cdot F=\bigcup_{E\in\mathcal{E}}\{E\cdot F\}$.

\section{Derivation \emph{via} a Support}\label{sec deriv via}
We now introduce the notion of a derivation {\it via} a support.
Recall that a Brzozowski derivative is an expression, whereas an Antimirov derivative is a set of expressions.
In our 
framework,
a derivative is more generally an element of an arbitrary set, called the structure set. For example, a structure can be an expression, a set of expressions or a set of set of expressions.
A support is essentially made of a structure set equipped with operators,
and of a mapping that transforms a structure into an expression.

  \begin{definition}\label{def support}
    Let $\Sigma$ be an alphabet. Let $\mathbb{E}$ be a set and $\mathrm{h}$ be a mapping from $\mathbb{E}$ to $\mathrm{Exp}(\Sigma)$. Let $\mathcal{O}$ be a set containing:
    \begin{itemize}
      \item for any $k$-ary boolean function $\mathrm{f}$, an operator $\mathrm{f}_{\mathbb{E}}$ from $\mathbb{E}^k$ to $\mathbb{E}$,
      \item an operator $\cdot_\mathbb{E}$ from $\mathbb{E}\times \mathrm{Exp}(\Sigma)$ to $\mathbb{E}$.
    \end{itemize}
    Let $1_\mathbb{E}$ and $0_\mathbb{E}$ be two elements in $\mathbb{E}$. The $6$-tuple $\mathbb{S}=(\Sigma,\mathbb{E},\mathrm{h},\mathcal{O},1_\mathbb{E},0_\mathbb{E})$ is said to be a \emph{support} if the three following conditions are satisfied:
    \begin{enumerate}
      \item for any $k$ elements $\mathcal{E}_1,\ldots,\mathcal{E}_k$ in $\mathbb{E}$:
    
      \centerline{
        $\mathrm{h}(\mathrm{f}_{\mathbb{E}}(\mathcal{E}_1,\ldots,\mathcal{E}_k)) \sim \mathrm{f}_e(\mathrm{h}(\mathcal{E}_1),\ldots,\mathrm{h}(\mathcal{E}_k))$,
      }
    
      \item for any element $\mathcal{E}$ in $\mathbb{E}$, for any expression $E$ in $\mathrm{Exp}(\Sigma)$:
    
      \centerline{
        $\mathrm{h}(\mathcal{E} \cdot_{\mathbb{E}} E ) \sim \mathrm{h}(\mathcal{E})\cdot E$,
      }
      
      \item $\mathrm{h}(1_\mathbb{E})\sim 1$ and $\mathrm{h}(0_\mathbb{E})\sim 0$.
    \end{enumerate}
  \end{definition}
  
  Notice that the expressions $\mathrm{h}(\mathrm{f}_{\mathbb{E}}(\mathcal{E}_1,\ldots,\mathcal{E}_k))$ and $\mathrm{f}_e(\mathrm{h}(\mathcal{E}_1),\ldots,\mathrm{h}(\mathcal{E}_k))$ need not to be identical. They are only required to define the same language.  
  A support is based on a set of generic structures that can be used to handle regular expressions. We now define the notion of regular expression derivation \emph{via} a support.
  
  \begin{definition}
    Let $\mathbb{S}=(\Sigma,\mathbb{E},\mathrm{h},\mathcal{O},1_\mathbb{E},0_\mathbb{E})$ be a support. The \emph{derivation via} $\mathbb{S}$ is the mapping $\mathrm{D}$ from $\Sigma^{+}\times\mathrm{Exp}(\Sigma)$ to $\mathbb{E}$ inductively defined for any  $a\in \Sigma$, for any word $w$ in $\Sigma^+$ and for any expression $E$ in $\mathrm{Exp}(\Sigma)$ by:

    \centerline{
        $\mathrm{D}(a,E)=
          \left\{
            \begin{array}{l@{\ }l}
              \mathrm{f}_{\mathbb{E}}(\mathrm{D}(a,E_1),\ldots,\mathrm{D}(a,E_k)) &  E=\mathrm{f}_e(E_1,\ldots,E_k),\\
             (\mathrm{D}(a,E_1) \cdot_{\mathbb{E}} E_2)  \vee_{\mathbb{E}} \mathrm{D}(a,E_2) &  E=E_1\cdot E_2\ \wedge\ \varepsilon\in L(E_1),\\
              \mathrm{D}(a,E_1) \cdot_{\mathbb{E}} E_2 &  E=E_1\cdot E_2\ \wedge\ \varepsilon\notin L(E_1),\\
              \mathrm{D}(a,E_1) \cdot_{\mathbb{E}} E_1^*  &  E=E_1^*,\\
              1_\mathbb{E} & E=a,\\
              0_\mathbb{E} &  E\in\Sigma\setminus\{a\}\cup\{1,0\}.\\
            \end{array}
          \right.$
    }      
     \centerline{   $\mathrm{D}(w,E)=
                    \mathrm{D}(u,\mathrm{h}(\mathrm{D}(a,E)))  \text{ if }w=au\ \wedge\ u\in\Sigma^+$\\
       }

  \end{definition}
   
     Furthermore, if for all expression $E$ in $\mathrm{Exp}(\Sigma)$, the set $\{D(w,E)\mid w\in\Sigma^{+}\}$ is finite, the derivation $\mathrm{D}$ is said to be a \emph{finite derivation}.

  \subsection{Classical Derivations are Derivations \emph{via}  a Support} 
  
     This subsection illustrates the fact that both Antimirov's derivation and  Brzozowski's one are derivations \emph{via} a support.
     
     \begin{definition}\label{def support anti}
       We denote by $\mathbb{S}_A=(\Sigma,\mathbb{E}=2^{\mathrm{Exp}(\Sigma)},\mathrm{h}_A,\mathcal{O}_A,\{1\},\emptyset)$ the $6$-tuple defined by:
   \begin{itemize}
     \item for any $\mathcal{E}\in 2^{\mathrm{Exp}(\Sigma)}$, $\mathrm{h}_A(\mathcal{E})=\sum_{E\in\mathcal{E}} E$,
     \item $\mathcal{O}_A=\{\mathrm{f}_{\mathbb{E}}\mid \mathrm{f}\text{ is a } k\text{-ary boolean function}\}\cup\{\cdot_{\mathbb{E}}\}$ where for any elements $\mathcal{E}_1,\ldots,\mathcal{E}_k$  in $2^{\mathrm{Exp}(\Sigma)}$,
     \begin{itemize}
        \item $\mathcal{E} \cdot_{\mathbb{E}} F=\bigcup_{E\in\mathcal{E}} \{E\cdot F\}$,
        \item $\mathrm{f}_{\mathbb{E}}(\mathcal{E}_1,\ldots,\mathcal{E}_k)=\mathcal{E}_1\cup \mathcal{E}_2$ if $k=2$ and $\mathrm{f}=\vee$,
        \item $\mathrm{f}_{\mathbb{E}}(\mathcal{E}_1,\ldots,\mathcal{E}_k)=\{\mathrm{f}_e(\mathrm{h}_A(\mathcal{E}_1),\ldots,\mathrm{h}_A(\mathcal{E}_k))\}$ otherwise.
      \end{itemize}
    \end{itemize}
     \end{definition}
   
   \begin{proposition}\label{prop supp anti est supp}
     The $6$-tuple $\mathbb{S}_A$ is a support. Furthermore, for any simple regular expression $E$ over $\Sigma$, for any symbol $a$, it holds $\mathrm{D}_A(a,E)=\frac{\partial}{\partial_a}(E)$, where  $\mathrm{D}_A$ is the derivation \emph{via} $\mathbb{S}_A$.
   \end{proposition}
   \begin{proof}
     Let $\mathcal{E}_1$ and $\mathcal{E}_2$ be two sets of simple regular expressions and $E$ be a simple regular expression. The condition \textbf{(1)} of Definition~\ref{def support} is satisfied since:
     
     \textbf{(a)} if $k=2$ and $\mathrm{f}=\vee$,
     
     \centerline{
       \begin{tabular}{l@{\ }l}
         $L(\mathrm{h_A}(\mathcal{E}_1\cup\mathcal{E}_2))$ & $=L(\sum_{E \in \mathcal{E}_1\cup\mathcal{E}_2} E)$\\
         & $=\bigcup_{E \in \mathcal{E}_1\cup\mathcal{E}_2} L(E)$\\
         & $=\bigcup_{E_1 \in \mathcal{E}_1} L(E_1) \cup \bigcup_{E_2 \in \mathcal{E}_2} L(E_2)$\\
         & $=L(\sum_{E_1 \in \mathcal{E}_1} E_1+\sum_{E_2 \in \mathcal{E}_2} E_2)$\\
         & $=L(\mathrm{h_A}(\mathcal{E}_1)+\mathrm{h_A}(\mathcal{E}_2))$
       \end{tabular}
       } 
     
     \textbf{(b)} and otherwise,
     
     \centerline{
       \begin{tabular}{l@{\ }l}
         $L(\mathrm{h_A}(\mathrm{f}_{\mathbb{E}}(\mathcal{E}_1,\ldots,\mathcal{E}_k))$ & $=L(\mathrm{h_A}(\{
         \mathrm{f}_{e}(\mathrm{h_A}(\mathcal{E}_1),\ldots,\mathrm{h_A}(\mathcal{E}_k))
         \}))$\\
         & $=L(
         \mathrm{f}_{e}(\mathrm{h_A}(\mathcal{E}_1),\ldots,\mathrm{h_A}(\mathcal{E}_k))
         )$\\
       \end{tabular}
       }  
       
     The condition \textbf{(2)} of Definition~\ref{def support} is satisfied since:
     
     \centerline{
       \begin{tabular}{l@{\ }l}
         $L(\mathrm{h}_A(\mathcal{E}_1\cdot_{\mathbb{E}} E))$ & $= L(\mathrm{h}_A(\bigcup_{E_1\in\mathcal{E}_1}\{E_1\cdot E\}))$\\
         & $= L(\sum_{E_1\in\mathcal{E}_1} E_1\cdot E)$\\
         & $= L(\sum_{E_1\in\mathcal{E}_1} E_1) \cdot L(E)$\\
         & $= L(\mathrm{h}_A(\mathcal{E}_1)) \cdot L(E)$\\
       \end{tabular}
       }

     The condition \textbf{(3)} of Definition~\ref{def support} is satisfied since:
     
     \centerline{
       \begin{tabular}{l@{\ }l}
         $L(\mathrm{h}_A(\{1\}))$ & $= L(1)$\\
       \end{tabular}
       } 
     
     \centerline{
       \begin{tabular}{l@{\ }l}
         $L(\mathrm{h}_A(\emptyset))$ & $=\emptyset=L(0)$\\
       \end{tabular}
       } 
       
     Consequently, $\mathbb{S}_A$ is a support. Moreover, since the operators $\cup$ and $\cdot_{\mathbb{E}}$ in $\mathcal{O}_A$ are the operations used in partial derivation, it can be shown by induction that for any simple regular expression $E$ over $\Sigma$ and for any symbol $a$, $\mathrm{D}_A(a,E)=\frac{\partial}{\partial_a}(E)$.
   \end{proof}
     
   \begin{definition}\label{def support brzo}
        We denote by $\mathbb{S}_B=(\Sigma,\mathbb{E}'=\mathrm{Exp}(\Sigma),\mathrm{h}_B,\mathcal{O}_B,1,0)$ the $6$-tuple defined by:
    \begin{itemize}
      \item for any $E \in \mathrm{Exp}(\Sigma)$, $\mathrm{h}_B(E)=E$,
      \item $\mathcal{O}_B=\{\mathrm{f}_{\mathbb{E}'}\mid \mathrm{f}\text{ is a } k\text{-ary boolean function}\}\cup\{\cdot_{\mathbb{E}'}\}$, where for any $E_1,\ldots,E_k$ elements in $\mathrm{Exp}(\Sigma)$,
      \begin{itemize}
        \item $E \cdot_{\mathbb{E}'} F=E\cdot F$,
        \item $\mathrm{f}_{\mathbb{E}'}(E_1,\ldots,E_k)=E_1 + E_2$ if $k=2$ and $\mathrm{f}=\vee$,
        \item $\mathrm{f}_{\mathbb{E}'}(E_1,\ldots,E_k)=\mathrm{f}_e(E_1,\ldots,E_k)$ otherwise.
      \end{itemize}
    \end{itemize}
   \end{definition}
   
   \begin{proposition}
     The $6$-tuple $\mathbb{S}_B$ is a support. Furthermore, for any regular expression $E$ over $\Sigma$, for any symbol $a$, it holds $\mathrm{D}_B(a,E)=\frac{d}{d_a}(E)$, where  $\mathrm{D}_B$ is the derivation \emph{via} $\mathbb{S}_B$.
   \end{proposition}
   \begin{proof}
     Since $\mathrm{h}_B$ is the identity and $\mathbb{E}'=\mathrm{Exp}(\Sigma)$, it is obvious that $\mathbb{S}_B$ is a support. Moreover, since the operators in $\mathcal{O}_B$ are the operators of regular expressions, it can be shown by induction that for any regular expression $E$ over $\Sigma$ and for any symbol $a$, $\mathrm{D}_B(a,E)=\frac{d}{d_a}(E)$.
   \end{proof}
   
    Let us notice that, by definition, the derivation $\mathrm{D}_A$ more generally addresses  unrestricted expressions ; therefore it provides a natural extension for Antimirov derivation. See~\cite{CCM11b} and Section~\ref{sec deriv via set clausal form} for alternative extensions.

 \subsection{Main Properties of Supports}
  
We now show that the language denoted by the expression associated with any derivative $D(w,E)$ is equal to the corresponding quotient. 

 \begin{proposition}\label{preservation du langage}
   Let $\mathrm{D}$ be the derivation via a support $\mathbb{S}=(\Sigma,\mathbb{E},\mathrm{h},\mathcal{O},1_\mathbb{E},$ $0_\mathbb{E})$. Then for any word $w$ in $\Sigma^{+}$, for any expression $E$ in $\mathrm{Exp}(\Sigma)$, it holds:
   
   \centerline{$L(\mathrm{h}(\mathrm{D}(w,E)))=w^{-1}(L(E))$.}
 \end{proposition}
 \begin{proof}
 
   By recurrence over the length of $w$.
   
   \begin{enumerate}
     \item Let $w=a\in\Sigma$. By induction over the structure of $E$.

    \centerline{
  $L(\mathrm{h}(\mathrm{D}(a,a)))=L(\mathrm{h}(1_\mathbb{E}))=\{\varepsilon\}=a^{-1}(L(a))$\\
    }
\medskip

    \centerline{
     \begin{tabular}{l@{\ }l}
      &$L(\mathrm{h}(\mathrm{D}(a,b)))=L(\mathrm{h}(\mathrm{D}(a,1)))=L(\mathrm{h}(\mathrm{D}(a,0)))$\\
      $=$& $L(\mathrm{h}(0_\mathbb{E}))=\emptyset=a^{-1}(L(b))=a^{-1}(L(1))=a^{-1}(L(0))$\\
     \end{tabular}
    }

 \medskip 
 
    \centerline{
     \begin{tabular}{l@{\ }l}
     & $L(\mathrm{h}(\mathrm{D}(a,E_1^*)))=L(\mathrm{h}(\mathrm{D}(a,E_1) \cdot_{\mathbb{E}} E_1^*))$\\
       $=$&$L(\mathrm{h}(\mathrm{D}(a,E_1)))\cdot L(E_1^*)=a^{-1}(L(E_1))\cdot L(E_1^*)=a^{-1}(L(E_1^*))$\\
     \end{tabular}
    }
  \medskip

       \centerline{
     \begin{tabular}{llll}
      &$L(\mathrm{h}(\mathrm{D}(a,\mathrm{f}_e(E_1,\ldots,E_k))))=L(\mathrm{h}(\mathrm{f}_\mathbb{E}(\mathrm{D}(a,E_1),\ldots,\mathrm{D}(a,E_k))))$\\
       $=$&$\mathrm{f}_L(L(\mathrm{h}(\mathrm{D}(a,E_1))),\ldots,L(\mathrm{h}(\mathrm{D}(a,E_k))))$\\
        $=$&$\mathrm{f}_L(a^{-1}(L(E_1)),\ldots,a^{-1}(L(E_k)))$\\
       $=$&$a^{-1}(\mathrm{f}_L(L(E_1),\ldots,L(E_k)))=a^{-1}(L(\mathrm{f}_e(E_1,\ldots,E_k)))$\\\
     \end{tabular}
    }
\medskip
    
    Let us consider that $\varepsilon\in L(E_1)$:

\centerline{
       \begin{tabular}{lll}
      &$L(\mathrm{h}(\mathrm{D}(a,E_1 \cdot E_2)))=L(\mathrm{h}((\mathrm{D}(a,E_1) \cdot_\mathbb{E} E_2) +_\mathbb{E} \mathrm{D}(a,E_2)))$\\
      $=$&$L(\mathrm{h}(\mathrm{D}(a,E_1) \cdot_\mathbb{E} E_2))\cup  L(\mathrm{h}(\mathrm{D}(a,E_2)))$\\
      $=$&$L(\mathrm{h}(\mathrm{D}(a,E_1))) \cdot  L(E_2)\cup  L(\mathrm{h}(\mathrm{D}(a,E_2)))$\\
      $=$&$a^{-1}(L(E_1)) \cdot  L(E_2)\cup  a^{-1}(L(E_2))$\\
      $=$&$a^{-1}(L(E_1) \cdot  L(E_2))=a^{-1}(L(E_1 \cdot E_2))$\\
\end{tabular}}
\medskip
    
    Let us consider that $\varepsilon\notin L(E_1)$:

\begin{center}
     \begin{tabular}{llll}
     &$L(\mathrm{h}(\mathrm{D}(a,E_1 \cdot E_2)))$&$=$&$L(\mathrm{h}(\mathrm{D}(a,E_1) \cdot_\mathbb{E} E_2))$\\
     $=$&$L(\mathrm{h}(\mathrm{D}(a,E_1))) \cdot  L(E_2)$&$=$&$a^{-1}(L(E_1)) \cdot  L(E_2)$\\
       $=$&$a^{-1}(L(E_1) \cdot  L(E_2))$&$=$&$a^{-1}(L(E_1 \cdot E_2))$\\
       \end{tabular}
    \end{center}   

     \item Let $w=au$ with $a\in\Sigma$ and $u\in\Sigma^{+}$. According to the recurrence hypothesis, 

     \begin{center}
       \begin{tabular}{llll}
         &$L(\mathrm{h}(\mathrm{D}(w,E)))$&$=$&$L(\mathrm{h}(\mathrm{D}(u,\mathrm{h}(\mathrm{D}(a,E)))))$\\
          $=$&$u^{-1}(L(\mathrm{h}(\mathrm{D}(a,E))))$&$=$&$u^{-1}(a^{-1}(L(E)))=(au)^{-1}(L(E))$\\
       \end{tabular}
     \end{center}

   \end{enumerate}
   
   \end{proof}
 
From Proposition \ref{preservation du langage} we deduce that $r_{\varepsilon}(h(D(w,E)))=r_{\varepsilon}(w^{-1}L(E))$. This property does not depend whether the derivation is finite or not
and since the boolean $r_{\varepsilon}(E)$ can be inductively computed for any regular expression $E$, any support defines a syntactical solution of the membership problem of the language $L(E)$.

  \begin{corollary}
    For a given regular expression $E$, any derivation via a support can be used to solve the membership problem for $L(E)$.
  \end{corollary}
  
As an example, the support $\mathbb{S}_B$ of Definition~\ref{def support brzo} can be used to solve the membership test, even if the associated derivation is not finite. 
  
Given an expression, the finiteness of the set of its derivatives
is a necessary condition for the construction of an associated finite automaton.
It is well-known that the set of Brzozowski's derivatives is not necessarily finite whereas the set of
dissimilar derivatives
and the set of 
Antimirov's derived
terms are finite sets.
We now give two sufficient conditions of finiteness in the general case. The first one has already be stated in the case of Brzozowski derivatives~\cite{Brz64}. The second one is related to the mapping $\mathrm{h}$ of the support, that needs to satisfy specific properties. 

    \begin{proposition}\label{prop finitude}
    Let $\mathrm{D}$ be the derivation via a support $\mathbb{S}=(\Sigma,\mathbb{E},\mathrm{h},\mathcal{O},1_\mathbb{E},0_\mathbb{E})$. The following set of conditions is sufficient for the mapping $\mathrm{D}$ to be a finite derivation:
    
      \begin{enumerate}
        \item $\vee_{\mathbb{E}}$ is associative, commutative and idempotent (\textbf{H1});
        \item for any $k$-ary boolean function $\mathrm{f}$, for any $k$ elements $\mathcal{E}_1,\ldots,\mathcal{E}_k$ in $\mathbb{E}$,
        
        \centerline{
          $\mathrm{D}(a,\mathrm{h}(\mathrm{f}_{\mathbb{E}}(\mathcal{E}_1,\ldots,\mathcal{E}_k)))=\mathrm{f}_{\mathbb{E}}(\mathrm{D}(a,\mathrm{h}(\mathcal{E}_1)),\ldots,\mathrm{D}(a,\mathrm{h}(\mathcal{E}_k)))$ (\textbf{H2}).
        }
      \end{enumerate}
      
  \end{proposition}
  \begin{proof}  
    Let us show by induction over the structure of regular expressions that for any expression $E$ in $\mathrm{Exp}(\Sigma)$, the set $\{\mathrm{D}(w,E)\mid w\in\Sigma^{+}\}$ is finite.

     If  $\mathbf{E\in\Sigma\cup\{1,0\}}$:      
        According to the definition of derivation, the proposition holds. \\

       If $\mathbf{E=E_1^{*}}$:
        Let $w$ be a word in $\Sigma^{+}$. 
        Let us show by recurrence over the length of $w$ that $\mathrm{D}(w,E_1^{*})$ is a finite $\vee_{\mathbb{E}}$-combination of elements in the set $\{\mathrm{D}(w',E_1) \cdot_{\mathbb{E}} E_1^* \mid w'\neq\varepsilon\text{ is a suffix of }w\}$.
          \begin{enumerate}
            \item Let $w=a\in\Sigma$. Since $\mathrm{D}(a,E_1^{*})=\mathrm{D}(a,E_1) \cdot_{\mathbb{E}} E_1^*$, the property holds.
            \item Let $w=ua$ with $a\in\Sigma$ and $u\in\Sigma^{+}$. By definition, $\mathrm{D}(ua,E_1^{*})=\mathrm{D}(a,\mathrm{h}(\mathrm{D}(u,E_1^{*})))$. By the recurrence hypothesis, $\mathrm{D}(u,E_1^{*})$ is a finite $\vee_{\mathbb{E}}$-combination of elements in the set $\{\mathrm{D}(w',E_1) \cdot_{\mathbb{E}} E_1^* \mid w'\neq\varepsilon \text{ is a suffix of }u\}$. According to hypothesis \textbf{H2}, $\mathrm{D}(a,\mathrm{h}(\mathrm{D}(u,E_1^{*})))$ is a finite $\vee_{\mathbb{E}}$-combination of elements in the set $\{\mathrm{D}(a,\mathrm{h}(\mathrm{D}(w',E_1^{*}))) \cdot_{\mathbb{E}} E_1^*\mid w'\neq\varepsilon\text{ is a suffix of }u\}\cup \{\mathrm{D}(a,E_1) \cdot_{\mathbb{E}} E_1^* \}$. So, $\mathrm{D}(ua,E_1^{*})$ is a finite $\vee_{\mathbb{E}}$-combination of elements in the set $\{\mathrm{D}(w',E_1) \cdot_{\mathbb{E}} E_1^*\mid w'\neq\varepsilon\text{ is a suffix of }ua\}$.
          \end{enumerate}

        As a consequence, since the set $\{\mathrm{D}(w,E_1)\mid w\in\Sigma^{+}\}$ is a finite set by induction hypothesis, since $\mathrm{Card}(\bigcup_{w\in\Sigma^+}\{\mathrm{D}(w',E_1)\mid w'\neq\varepsilon$ is a suffix of  $w\})=\mathrm{Card}(\bigcup_{w\in\Sigma^+}\{\mathrm{D}(w,E_1)\})$ and since $\vee_{\mathbb{E}}$ is associative, commutative and idempotent, we get:
          
          \centerline{
            $\mathrm{Card}(\{\mathrm{D}(w,E_1^{*})\mid w\in\Sigma^{+}\})\leq 2^{\mathrm{Card}(\{\mathrm{D}(w,E_1)\mid w\in\Sigma^{+}\})}$.
          } 
          
\ \\

    If $\mathbf{E=\mathrm{f}_e(E_1,\ldots,E_k)}$:      
        Let $w$ be a word in $\Sigma^{+}$.
         Let us show by recurrence over the length of $w$ that $\mathrm{D}(w,\mathrm{f}_e(E_1,\ldots,E_k))=\mathrm{f}_{\mathbb{E}}(\mathrm{D}(w,E_1),$ $\ldots,\mathrm{D}(w,E_k))$.
          
          \begin{enumerate}
            \item Let $w=a\in\Sigma$. According to the definition of $\mathrm{D}$, the property holds.
            \item Let $w=ua$ with $a\in\Sigma$ and $u\in\Sigma^{+}$. By definition, $\mathrm{D}(ua,\mathrm{f}_e(E_1,$ $\ldots,E_k))$ $=\mathrm{D}(a,\mathrm{h}(\mathrm{D}(u,\mathrm{f}_e(E_1,\ldots,E_k))))$. By the recurrence hypothesis, $\mathrm{D}(u,\mathrm{f}_e(E_1,$ $\ldots,E_k))=\mathrm{f}_{\mathbb{E}}(\mathrm{D}(u,E_1),\ldots,\mathrm{D}(u,E_k))$. According to hypothesis \textbf{H2},
            
            \centerline{ $\mathrm{D}(a,\mathrm{h}(\mathrm{f}_{\mathbb{E}}(\mathrm{D}(u,E_1),\ldots,\mathrm{D}(u,E_k))))$ $=\mathrm{f}_{\mathbb{E}}(\mathrm{D}(a,\mathrm{h}(\mathrm{D}(u,$ $E_1))),\ldots,\mathrm{D}(a,\mathrm{h}(\mathrm{D}(u,E_k))))$}
            
            \centerline{ $=\mathrm{f}_{\mathbb{E}}(\mathrm{D}(ua,E_1),\ldots,\mathrm{D}(ua,E_k))$.}
            
          \end{enumerate}  
          
         As a consequence, since for all integer $j$ in $\{1,\ldots,k\}$ the set $\{\mathrm{D}(w,E_j)\mid w\in\Sigma^{+}\}$ is finite by induction hypothesis, we get:
       \begin{center}
        $\mathrm{Card}(\{\mathrm{D}(w,\mathrm{f}_e(E_1,\ldots,E_k))\mid w\in\Sigma^{+}\})$\\
        $\leq
         \mathrm{Card}(\{\mathrm{D}(w,E_1)\mid w\in\Sigma^{+}\})\times \cdots \times \mathrm{Card}(\{\mathrm{D}(w,E_k)\mid w\in\Sigma^{+}\})$.
      \end{center}
    
    If $\mathbf{E=E_1\cdot E_2}$:
        Let $w$ be a word in $\Sigma^{+}$.
          Let us show by recurrence over the length of $w$ that either $\mathrm{D}(w,E_1\cdot E_2)=\mathrm{D}(w,E_1) \cdot_{\mathbb{E}} E_2 $ or $\mathrm{D}(w,E_1\cdot E_2)=(\mathrm{D}(w,E_1) \cdot_{\mathbb{E}} E_2) \vee_{\mathbb{E}} \mathcal{E}$ where $\mathcal{E}$ is a finite $\vee_{\mathbb{E}}$-combination of elements in the set $\{\mathrm{D}(w',E_2)\mid w'\neq\varepsilon \text{ is a suffix of }w\}$.
          
          \begin{enumerate}
            \item Let $w=a\in\Sigma$. According to the definition of $\mathrm{D}$, the property holds.
            \item Let $w=ua$ with $a\in\Sigma$ and $u\in\Sigma^{+}$.
            
            By definition, $\mathrm{D}(ua,E_1\cdot E_2)=\mathrm{D}(a,\mathrm{h}(\mathrm{D}(u,E_1\cdot E_2)))$.
            
            Two cases have to be considered:
             
             \begin{enumerate}
               \item $\mathrm{D}(u,E_1\cdot E_2)=\mathrm{D}(u,E_1) \cdot_{\mathbb{E}} E_2$. Either $\mathrm{D}(ua,E_1\cdot E_2)=\mathrm{D}(a,\mathrm{h}(\mathrm{D}(u,$ $E_1) \cdot_{\mathbb{E}} E_2))$, or $\mathrm{D}(ua,E_1\cdot E_2)=(\mathrm{D}(a,\mathrm{h}(\mathrm{D}(u,E_1) \cdot_{\mathbb{E}} E_2 )))  \vee_{\mathbb{E}} \mathrm{D}(a,\mathrm{h}(E_2)) $. According to the recurrence hypothesis, both of these cases satisfy the proposition.
               
               \item $\mathrm{D}(u,E_1\cdot E_2)=(\mathrm{D}(u,E_1) \cdot_{\mathbb{E}} E_2 ) \vee_{\mathbb{E}} \mathcal{E} $ where $\mathcal{E}$ is a finite $\vee_{\mathbb{E}}$-combination of elements in the set $\{\mathrm{D}(w',E_2)\mid w'\neq\varepsilon\text{ is a suffix of }u\}$. Either $\mathrm{D}(ua,E_1\cdot E_2)=\mathrm{D}(a,\mathrm{h}(\mathrm{D}(u,E_1) \cdot_{\mathbb{E}} E_2 )) \vee_{\mathbb{E}} \mathrm{D}(a,\mathrm{h}(\mathcal{E})) $ or $\mathrm{D}(ua,E_1\cdot E_2)=\mathrm{D}(a,\mathrm{h}(\mathrm{D}(u,E_1) \cdot_{\mathbb{E}} E_2 ))               
               \vee_{\mathbb{E}} \mathrm{D}(a,\mathrm{h}(E_2))               
               \vee_{\mathbb{E}} \mathrm{D}(a,\mathrm{h}(\mathcal{E}))$. According to hypothesis \textbf{H2}, $\mathcal{E}'=\mathrm{D}(a,\mathrm{h}(\mathcal{E}))$ is a finite $\vee_{\mathbb{E}}$-combination of elements in the set $\{\mathrm{D}(a,\mathrm{h}(\mathrm{D}(w',E_2)))\mid w'\neq\varepsilon \text{ is a suffix of }u\}$, set that equals $\{\mathrm{D}(w',E_2)\mid w'\neq\varepsilon$ is a suffix of $ua\}$.
             \end{enumerate} 
             
             Consequently, $\mathrm{D}(ua,E_1\cdot E_2)=(\mathrm{D}(ua,E_1) \cdot_{\mathbb{E}} E_2 ) \vee_{\mathbb{E}}   \mathcal{E} $ where $\mathcal{E}$ is a finite $\vee_{\mathbb{E}}$-combination of elements in the set $\{\mathrm{D}(w',E_2)\mid w'\neq\varepsilon\text{ is a suffix of }ua\}$. 
          \end{enumerate} 
          
          As a consequence, since the sets $\{\mathrm{D}(w,E_1)\mid w\in\Sigma^{+}\}$ and $\{\mathrm{D}(w,E_2)\mid w\in\Sigma^{+}\}$ are finite by induction hypothesis and since $\vee_{\mathbb{E}}$ is associative, commutative and idempotent,  we get:
          
         \begin{center}
            $\mathrm{Card}(\{\mathrm{D}(w,E_1\cdot E_2)\mid w\in\Sigma^{+}\})$\\
            $\leq \mathrm{Card}(\{\mathrm{D}(w,E_1)\mid w\in\Sigma^{+}\})\times 2^{\mathrm{Card}(\{\mathrm{D}(w,E_2)\mid w\in\Sigma^{+}\})}$.
         \end{center}
         
  \end{proof}
  

    The derivation $\mathrm{D}_A$ of Definition~\ref{def support anti} is an example of finite derivation since $\cup$ is associative, commutative and idempotent, and since for any $k$-ary boolean function $\mathrm{f}$ and for any $k$ elements $\mathcal{E}_1,\ldots,\mathcal{E}_k$ in $2^{\mathrm{Exp}(\Sigma)}$, we have:
        
        \centerline{
          $\mathrm{D}_A(a,\mathrm{h}_A(\mathrm{f}_{\mathbb{E}}(\mathcal{E}_1,\ldots,\mathcal{E}_k)))=\mathrm{f}_{\mathbb{E}}(\mathrm{D}_A(a,\mathrm{h}_A(\mathcal{E}_1)),\ldots,\mathrm{D}_A(a,\mathrm{h}_A(\mathcal{E}_k)))$.
        }
      
    On the opposite, since $+$ is not an ACI law, Proposition~\ref{prop finitude} does not allow us to conclude for the derivation $\mathrm{D}_B$ of Definition~\ref{def support brzo}. Brzozowski showed in~\cite{Brz64} that it is possible to compute a finite set of dissimilar derivatives using a quotient of the expressions  w.r.t. an ACI sum. It can be achieved by considering the support $\mathbb{S}'_B$ defined as follows:
    
    \begin{definition}\label{def support brzo aci}
      We denote by $\mathbb{S}'_B=(\Sigma,\mathbb{E}''=2^{\mathrm{Exp}(\Sigma)},\mathrm{h}_A,\mathcal{O}'_B,\{1\},\{0\})$ the $6$-tuple defined by:
    \begin{itemize}
      \item for any $\mathcal{E}\in 2^{\mathrm{Exp}(\Sigma)}$, $\mathrm{h}_A(\mathcal{E})=\sum_{E\in\mathcal{E}} E$ (Definition~\ref{def support anti}),
      \item $\mathcal{O}'_B=\{\mathrm{f}_{\mathbb{E}''}\mid \mathrm{f}\text{ is a } k\text{-ary boolean function}\}\cup\{\cdot_{\mathbb{E}''}\}$, where for any $\mathcal{E}_1,\ldots,\mathcal{E}_k$ elements in $\mathbb{E}''$,
      \begin{itemize}
        \item $\mathcal{E}_1 \cdot_{\mathbb{E}''} F=\{(\sum_{E\in\mathcal{E}_1} E)\cdot F\}$,
        \item $\mathrm{f}_{\mathbb{E}''}(\mathcal{E}_1,\ldots,\mathcal{E}_k)=\mathcal{E}_1\cup \mathcal{E}_2$ if $k=2$ and $\mathrm{f}=\vee$,
        \item $\mathrm{f}_{\mathbb{E}''}(\mathcal{E}_1,\ldots,\mathcal{E}_k)=\{\mathrm{f}_e(\mathrm{h}_A(\mathcal{E}_1),\ldots,\mathrm{h}_A(\mathcal{E}_k))\}$ otherwise.
      \end{itemize}
    \end{itemize}
    \end{definition}
    
    \begin{proposition}
      The $6$-tuple $\mathbb{S}'_B$ is a support. Furthermore, for any regular expression $E$ and for any symbol $a$, it holds that $\mathrm{h}_A(\mathrm{D}'_{B}(a,E))$ is the dissimilar derivative of $E$ w.r.t. $a$, where $\mathrm{D}'_{B}$ is the derivation \emph{via} the support $\mathbb{S}'_B$.
    \end{proposition}
    \begin{proof}
    According to Definition~\ref{def support anti}, Definition~\ref{def support brzo aci} and Proposition~\ref{prop supp anti est supp}, the conditions \textbf{(1)} and \textbf{(3)} of Definition~\ref{def support} are satisfied by $\mathbb{S}'_B$.    
    
     The condition \textbf{(2)} of Definition~\ref{def support} is satisfied since:
     
     \centerline{
       \begin{tabular}{l@{\ }l}
         $L(\mathrm{h}_A(\mathcal{E}_1 \cdot_{\mathbb{E}''} E))$ & $= L(\mathrm{h}_A(\{(\sum_{F\in\mathcal{E}_1} F)\cdot E\}))$\\
         & $= L((\sum_{F\in\mathcal{E}_1} F)\cdot E)$\\
         & $= L(\sum_{F\in\mathcal{E}_1} F)\cdot L(E)$\\
         & $= L(\mathrm{h}_A(\mathcal{E}_1))\cdot L(E)$\\
       \end{tabular}
       } 
       
       Consequently, $\mathbb{S}'_B$ is a support. 
       
       Moreover, since the operator $\cup$ is an ACI law and since the catenation product $\cdot_{\mathbb{E}''}$ returns a singleton, it can be shown by induction that for any regular expression $E$ and for any symbol $a$, it holds that $\mathrm{h}_A(\mathrm{D}'_{B}(a,E))$ is the dissimilar derivative of $E$ w.r.t. $a$.
    \end{proof}
  
\section{From Derivation \emph{via} a Support to Automata}\label{sec from der to aut}

  Computing the set of derivatives of a regular expression $E$ w.r.t. a derivation $\mathrm{D}$ is similar as computing the transition function $\delta$ of an automaton, where $\delta(E,w)=\mathrm{h}(\mathrm{D}(w,E))$. As far as alternating automata are concerned, the resulting expression $\mathrm{h}(\mathrm{D}(w,E))$ needs to be transformed into a boolean formula. This computation is performed through a \emph{base function} defined as follows.

  \begin{definition}\label{def base}
    Let $\Sigma$ be an alphabet. A \emph{base function} $\mathrm{B}$ is a mapping from $\mathrm{Exp}(\Sigma)$ to $\mathrm{BoolForm}(\mathrm{Exp}(\Sigma))$ such that for any expression $E$ and for any word $w$ in $\Sigma^{*}$:
    
    \centerline{$w\in L(E)$ $\Leftrightarrow$ $\mathrm{eval}_{r_w}(B(E))=1$.}
  \end{definition}
  
  \begin{definition}
    Let $\mathrm{B}$ be a base function and $\mathrm{D}$ be the derivation \emph{via} the support $\mathbb{S}=(\Sigma,\mathbb{E},\mathrm{h},\mathcal{O},1_\mathbb{E},0_\mathbb{E})$. Let $E$ be an expression in $\mathrm{Exp}(\Sigma)$. Let $A=(\Sigma,Q,I,F,\delta)$ be the automaton defined by:
    
    \begin{itemize}
      \item $Q=\mathrm{Exp}(\Sigma)$,
      \item $I=E$,
      \item $\forall q\in Q$, 
        $F(q)=
          \left\{
            \begin{array}{l@{\ }l}
              1 & \text{ if }\varepsilon\in L(q),\\
              0 & \text{ otherwise,}
            \end{array}
          \right.
        $
      \item $\forall a\in\Sigma, \forall q\in Q, \delta(q,a)=\mathrm{B}(\mathrm{h}(\mathrm{D}(a,q)))$.
    \end{itemize}
    
    The accessible part of $A$ is said to be the $(\mathrm{D},\mathrm{B})$-\emph{alternating automaton of} $E$. Notice that there may exist an infinite number of states.
  \end{definition}
  
  \begin{theorem}~\label{thm aa lang}
     The $(\mathrm{D},\mathrm{B})$-alternating automaton of a regular expression $E$ recognizes $L(E)$.
  \end{theorem}
  \begin{proof}
  Let $A=(\Sigma,Q,I,F,\delta)$ be the $(\mathrm{D},\mathrm{B})$-alternating automaton of $E$. 
    Let $\mathrm{D}$ be the derivation via the support $\mathbb{S}=(\Sigma,\mathbb{E},\mathrm{h},\mathcal{O},1_{\mathbb{E}},0_{\mathbb{E}})$. Let $w$ be a word in $\Sigma^{+}$. 
       Let us show by recurrence over the length of $w$ that for any boolean formula $\phi=\mathrm{f}_{\mathbb{B}}(q_1,\ldots,q_k)$ in $\mathrm{FormBool}(Q)$, the following \textbf{(P)} proposition holds:
      
        \centerline{
          $\mathrm{eval}_{F}(\delta(\phi,w))=\mathrm{eval}_{r_{\varepsilon}}(\mathrm{f}_{\mathbb{B}}(\mathrm{B}(\mathrm{h}(\mathrm{D}(w,q_1))),\ldots,\mathrm{B}(\mathrm{h}(\mathrm{D}(w,q_k))))$.
        }
        
        \begin{enumerate}
          \item If $w=a\in\Sigma$, by definition of the transition function $\delta$,          
           $\delta(\phi,a)$ $=\mathrm{f}_{\mathbb{B}}(\delta(q_1,a),$ $\ldots,\delta(q_k,a))$. By construction, for any integer $j$ in $\{1,\ldots,k\}$, $\delta(q_j,a)=\mathrm{B}(\mathrm{h}(\mathrm{D}(a,q_j)))$.
          
           Since for any state $q\in Q$, $F(q)=1\Leftrightarrow \varepsilon\in L(q)$, the following proposition holds:   $\mathrm{eval}_{F}(\delta(\phi,a))=\mathrm{eval}_{r_{\varepsilon}}(\mathrm{f}_{\mathbb{B}}(\mathrm{B}(\mathrm{h}(\mathrm{D}(a,q_1))),\ldots,\mathrm{B}(\mathrm{h}(\mathrm{D}(a,q_k))))$.
          
          \item Let $w=au$ with $a\in\Sigma$ and $u\in\Sigma^+$. Then it holds:
          
              \begin{tabular}{l@{\ }l}
                & $\mathrm{eval}_{F}(\delta(\phi,au))$\\
                & $=\mathrm{eval}_{F}(\delta(\delta(\phi,a),u))$ \hfill\textbf{(Definition of $\delta$)}\\
                & $=\mathrm{eval}_{F}(\delta(\mathrm{f}_{\mathbb{B}}(\delta(q_1,a),\ldots,\delta(q_k,a))),u))$ \hfill\textbf{(Definition of $\delta(\phi,a)$)}\\
                & $=\mathrm{eval}_{F}(\mathrm{f}_{\mathbb{B}}(\delta(\delta(q_1,a),u),\ldots,\delta(\delta(q_k,a),u)))$ \hfill\textbf{(Definition of $\delta$)}\\
                & $=\mathrm{eval}_{F}(\mathrm{f}_{\mathbb{B}}(\delta(\mathrm{B}(\mathrm{h}(\mathrm{D}(a,q_1))),u),\ldots,\delta(\mathrm{B}(\mathrm{h}(\mathrm{D}(a,q_k))),u)))$\\ & \hfill\textbf{(Construction of $\delta$)}\\
                & $=\mathrm{f}(\mathrm{eval}_{F}(\delta(\mathrm{B}(\mathrm{h}(\mathrm{D}(a,q_1))),u)),\ldots,\mathrm{eval}_{F}(\delta(\mathrm{B}(\mathrm{h}(\mathrm{D}(a,q_k))),u)))$ \\ & \hfill\textbf{(Definition of $\mathrm{eval}_F$)}\\
                & $=\mathrm{f}(\mathrm{eval}_{r_{\varepsilon}}(\mathrm{B}(\mathrm{h}(\mathrm{D}(u,\mathrm{h}(\mathrm{D}(a,q_1)))))),$\\
                &\ \ \ \ \ \ \ \ \ \ \ \ \ \ \ \ \ \ \ \ $\ldots,\mathrm{eval}_{r_{\varepsilon}}(\mathrm{B}(\mathrm{h}(\mathrm{D}(u,\mathrm{h}(\mathrm{D}(a,q_k)))))))$\\ &  \hfill\textbf{(Induction hypothesis)}\\
                & $=\mathrm{eval}_{r_{\varepsilon}}(\mathrm{f}_{\mathbb{B}}(\mathrm{B}(\mathrm{h}(\mathrm{D}(u,\mathrm{h}(\mathrm{D}(a,q_1))))),\ldots,\mathrm{B}(\mathrm{h}(\mathrm{D}(u,\mathrm{h}(\mathrm{D}(a,q_k)))))))$\\ &  \hfill\textbf{(Definition of $\mathrm{eval}_F$)}\\
                & $=\mathrm{eval}_{r_{\varepsilon}}(\mathrm{f}_{\mathbb{B}}(\mathrm{B}(\mathrm{h}(\mathrm{D}(au,q_1))),\ldots,\mathrm{B}(\mathrm{h}(\mathrm{D}(au,q_k))))$ \hfill\textbf{(Definition of $\mathrm{D}$)}\\
              \end{tabular}
        \end{enumerate}
        
        Finally, it holds:
        
          \begin{tabular}{l@{\ }l}
            $w\in L(E)$ & $\Leftrightarrow$ $\varepsilon\in L(\mathrm{h}(\mathrm{D}(w,E)))$ \hfill\textbf{(Proposition~\ref{preservation du langage})}\\
            & $\Leftrightarrow$ $\mathrm{eval}_{r_{\varepsilon}}(\mathrm{B}((\mathrm{h}(\mathrm{D}(w,E)))))=1$ \hfill\textbf{(Definition~\ref{def base})}\\
            & $\Leftrightarrow$ $\mathrm{eval}_{F}(\delta(E,w))=1$ \hfill\textbf{(Proposition \textbf{(P)})}\\
            & $\Leftrightarrow$ $w\in L(A)$\ \ \ \ \ \hfill\textbf{(Definition of the language of an AA)}\\
          \end{tabular}

  \end{proof}

  Previous definitions and properties address non necessarily finite automata. We now give sufficient conditions for the finiteness of the automaton.
  The basic idea is that if it is equivalent to derive an expression $E$ or to derive its atoms $\mathrm{Atom}(\mathrm{B}(E))$, then the set of atoms obtained by a finite derivation is finite. 
  
  \begin{definition}
    Let $\mathbb{S}=(\Sigma,\mathbb{E},\mathrm{h},\mathcal{O},1_\mathbb{E},0_\mathbb{E})$ be a support, let $\mathrm{D}$ be the derivation \emph{via} $\mathbb{S}$ and $\mathrm{B}$ be a base function. The couple $(\mathrm{D},\mathrm{B})$ satisfies the \emph{atom-derivability property} if for any expression $E$ in $\mathrm{Exp}(\Sigma)$ and for any symbol $a$ in $\Sigma$:
    
    \centerline{
      $\mathrm{Atom}(\mathrm{B}(\mathrm{h}(\mathrm{D}(a,E))))=\bigcup_{E'\in\mathrm{Atom}(\mathrm{B}(E))} \mathrm{Atom}(\mathrm{B}(\mathrm{h}(\mathrm{D}(a,E'))))$.
    }
  \end{definition}

  \begin{theorem}\label{thm afa}
    Let $A$ be the $(\mathrm{D},\mathrm{B})$-alternating automaton of a regular expression $E$. If $\mathrm{D}$ is finite and if $(\mathrm{D},\mathrm{B})$ satisfies the atom-derivability property, then:
    
    \centerline{
      $A$ is an AFA.
    }
  \end{theorem}
  \begin{proof}
    Let $\mathrm{D}$ be the derivation via the support $\mathbb{S}=(\Sigma,\mathbb{E},\mathrm{h},\mathcal{O},1_{\mathbb{E}},0_{\mathbb{E}})$. Let $w$ be a word in $\Sigma^+$.
    
    \begin{enumerate}
    
    \item  Let us show by recurrence over the length of $w$ that for any expression $q$ in $Q$, $\mathrm{Atom}(\mathrm{B}(\mathrm{h}(\mathrm{D}(w,q))))=\mathrm{Atom}(\delta(q,w))$.
  
  \begin{enumerate}
        
      \item If $w=a\in\Sigma$, $\delta(q,a)=\mathrm{B}(\mathrm{h}(\mathrm{D}(a,q)))$. Then, $\mathrm{Atom}(\mathrm{B}(\mathrm{h}(\mathrm{D}(a,q))))$ $=\mathrm{Atom}(\delta(q,a))$.
      
      \item If $w=ua$ with $a\in\Sigma$ and $u\in\Sigma^+$, by the recurrence hypothesis, $\mathrm{Atom}(\delta(q,u))=\mathrm{Atom}(\mathrm{B}(\mathrm{D}(u,q)))$.

         \begin{tabular}{l@{\ }l}
          & $\mathrm{Atom}(\delta(q,ua))$\\
          & $=\mathrm{Atom}(\delta(\delta(q,u),a))$ \hfill\textbf{(Definition of $\delta$)}\\
          & $=\mathrm{Atom}(\delta(\mathrm{f}_\mathbb{B}(q'_1,\ldots,q'_j),a))$ \\ & \hfill\textbf{(Definition of $\delta(q,u)$ with $\mathrm{Atom}(\delta(q,u))=\{q'_1,\ldots,q'_j\}$)}\\
          & $=\mathrm{Atom}(\mathrm{f}_\mathbb{B}(\delta(q'_1,a),\ldots,\delta(q'_j,a))$ \hfill\textbf{(Definition of $\delta$)}\\
          & $=\bigcup_{q'\in\{q'_1,\ldots,q'_j\}} \mathrm{Atom}(\delta(q',a))$ \hfill\textbf{(Definition of $\mathrm{Atom}$)}\\
          & $=\bigcup_{q' \in \mathrm{Atom}(\delta(q,u))} \mathrm{Atom}(\delta(q',a))$ \hfill\textbf{(Definition of $\{q'_1,\ldots,q'_j\}$)}\\
          & $=\bigcup_{q' \in \mathrm{Atom}(\mathrm{B}(\mathrm{h}(\mathrm{D}(u,q))))} \mathrm{Atom}(\mathrm{B}(\mathrm{h}(\mathrm{D}(a,q'))))$\\ &  \hfill\textbf{(Induction hypothesis and construction of $\delta$)}\\
          & $=\mathrm{Atom}(\mathrm{B}(\mathrm{h}(\mathrm{D}(a,\mathrm{h}(\mathrm{D}(u,q))))))$ \hfill\textbf{(Atom-derivability property)}\\
         \end{tabular}

    \end{enumerate}
    \item As a direct consequence of the previous point, since the set $\{\mathrm{D}(w,E)\mid w\in\Sigma^+\}$ is finite, so are the sets $\{\mathrm{B}(\mathrm{h}(\mathrm{D}(w,E)))\mid w\in\Sigma^+\}$ and $\bigcup_{w\in\Sigma^+} \mathrm{Atom}($ $\mathrm{B}(\mathrm{h}(\mathrm{D}(w,q))))$. Finally, the set $Q$, that is equal to $\bigcup_{w\in\Sigma^+} \mathrm{Atom}(\delta(q,w))$, is a finite set.
    \end{enumerate}
    \end{proof}
  
   \begin{example}
    Let $\mathrm{D}_A$ and $\mathrm{D}'_B$ be the derivations of Definition~\ref{def support anti} and Definition~\ref{def support brzo aci}.
    Let $\mathrm{B}_{A}$ be the base inductively defined for any expression by  $\mathrm{B}_{A}(E+F)=\mathrm{B}_{A}(E)\ \vee_{\mathbb{B}}\ \mathrm{B}_{A}(F)$, $\mathrm{B}_{A}(E)=E$ otherwise.
    Let $\mathrm{B}_{B}$ the base defined for any expression by $\mathrm{B}_{B}(E)=E$.
    It can be shown that any couple in $\{\mathrm{D}_A,\mathrm{D}'_B\}\times\{\mathrm{B}_{A},\mathrm{B}_{B}\}$ satisfies the atom-derivability property.
    
    Furthermore, the $(\mathrm{D}_A,\mathrm{B}_A)$-AFA of $E$ can be straightforwardly transformed into the 
    derived
    term NFA of $E$ defined by Antimirov in~\cite{Ant96}; the $(\mathrm{D}_A,\mathrm{B}_B)$-AFA of $E$ can be transformed into the partial derivative DFA of $E$ defined by Antimirov in~\cite{Ant96};  the $(\mathrm{D}'_B,\mathrm{B}_B)$-AFA of $E$ can be transformed into the dissimilar derivative DFA of $E$ defined by Brzozowski in~\cite{Brz64}. Notice that the $(\mathrm{D}'_B,\mathrm{B}_A)$-AFA of $E$ can be transformed into a NFA that is different from the NFA of Antimirov.
    
    Finally, the base $\mathrm{B}_C$ inductively defined by $\mathrm{B}_C(\mathrm{f}_e(E_1,\ldots,E_k)=\mathrm{f}_{\mathbb{B}}(\mathrm{B}_C(E_1),$ $\ldots,\mathrm{B}_C(E_k))$ for any operator $f_e$, $\mathrm{B}_C(E)=E$ otherwise, provides an AFA construction both from $\mathrm{D}_A$ and $\mathrm{D}'_B$.
  \end{example}

  \section{Derivation \emph{via} the Set of Clausal Forms}\label{sec deriv via set clausal form}
    
    In this section, we show that the set of clausal forms over the set of regular expressions and equipped with the right operators is a derivation support and that the associated $\mathrm{D}_{\mathbb{C}}$ derivation is finite. Furthermore we prove that the atom-derivability property is satisfied whenever the $\mathrm{D}_{\mathbb{C}}$ derivation is associated with a base function in the set $\{\mathrm{B}_A, \mathrm{B}_B, \mathrm{B}_C\}$. Finally we illustrate these results by the construction of the $(\mathrm{D}_{\mathbb{C}},\mathrm{B}_C)$-AFA of the expression $E=((ab)^*a)\mathrm{XOR}_e ((abab)^*a)$.
    
    \medskip
    
    Let us first recall some definitions about clausal forms and their operators.  
  A \emph{clausal form} over a set $X$ is an element in $\mathcal{C}(X)=2^{2^{X\cup \overline{X}}}$ where $\overline{X}=\{\overline{x}\mid x\in X\}$. Let $\ocup$ and $\ocap$ be the two mappings from $\mathcal{C}(X)\times \mathcal{C}(X)$ to $\mathcal{C}(X)$ and $\oneg$ be the function from $\mathcal{C}(X)$  to $\mathcal{C}(X)$ defined for any $\mathcal{C}_1,\mathcal{C}_2$ in $\mathcal{C}(X)$ by:
  
  \begin{itemize}
    \item $\mathcal{C}_1\ocup \mathcal{C}_2 =\mathcal{C}_1\cup \mathcal{C}_2$,
    \item $\mathcal{C}_1\ocap \mathcal{C}_2 =\bigcup_{(C_1,C_2)\in \mathcal{C}_1\times\mathcal{C}_2}\{C_1\cup C_2\}$, 
    \item   $ \oneg(\mathcal{C}_1) =\left\{
          \begin{array}{l@{\ }l}
            \{\emptyset\} & \text{ if } \mathcal{C}_1=\emptyset,\\
            \emptyset & \text{ if } \mathcal{C}_1=\{\emptyset\},\\
            \ocap_{C\in\mathcal{C}_1}  \ocup_{c\in C} \{\{\mathrm{n}(c)\} \} & \text{ otherwise,}\
          \end{array}
        \right.$  
  \end{itemize}
  
  where $\mathrm{n}(c) =
          \left\{
            \begin{array}{l@{\ }l}
              c' & \text{ if }c=\overline{c'},\\
              \overline{c} & \text{ otherwise.}
            \end{array}
          \right.$
  
  It can be shown that for any element $x\in X\cup \overline{X}$, for any clause $\mathcal{C}_1,\mathcal{C}_2,\mathcal{C}_3$ in $\mathcal{C}(X)$, the following conditions are satisfied:
  
  \begin{itemize}
    \item $\mathcal{C}_1 \ocap(\mathcal{C}_2\ocup \mathcal{C}_3)=(\mathcal{C}_1 \ocap \mathcal{C}_2)\ocup (\mathcal{C}_1 \ocap \mathcal{C}_3)$,
    \item $\oneg\oneg\{\{x\}\}=\{\{x\}\}$,
    \item $\oneg(\mathcal{C}_1\ocup\mathcal{C}_2)=\oneg(\mathcal{C}_1)\ocap \oneg(\mathcal{C}_2)$,
    \item $\oneg(\mathcal{C}_1\ocap\mathcal{C}_2)=\oneg(\mathcal{C}_1)\ocup \oneg(\mathcal{C}_2)$.
  \end{itemize}
  
  Furthermore, let us notice that there exist clauses $\mathcal{C}_1,\mathcal{C}_2,\mathcal{C}_3$ such that $\oneg\oneg\mathcal{C}_1\neq\mathcal{C}_1$ or $\mathcal{C}_1 \ocup(\mathcal{C}_2\ocap \mathcal{C}_3)=(\mathcal{C}_1 \ocup \mathcal{C}_2)\ocap (\mathcal{C}_1 \ocup \mathcal{C}_3)$.
  
  \medskip
    
    From now on, we will consider the set $\mathbb{C}=\mathcal{C}(\mathrm{Exp}(\Sigma))$ of the clausal forms over the set of regular expressions.    
    We now explain how a clausal form over the set of regular expressions is transformed into a regular expression. Let us consider the function $\mathrm{h}_{\mathbb{C}}$ defined from  $\mathbb{C}$ to $\mathrm{Exp}(\Sigma)$ for any element $\mathcal{C}$ in $\mathbb{C}$ by:
    
    \centerline{
      \begin{tabular}{l@{\ }l}
        $\mathrm{h}_{\mathbb{C}}(\mathcal{C})$ & $= 
          \left\{
          \begin{array}{l@{\ }l}
          0 & \text{ if } \mathcal{C}=\emptyset,\\          
          \sum_{C\in\mathcal{C}}
          \left\{
            \begin{array}{l@{\ }l}
              \neg_e 0 & \text{ if }C=\emptyset,\\
              \bigwedgee_{E\in C} c(E) & \text{ otherwise,}
            \end{array}
          \right. & \text{ otherwise.}
          \end{array}
          \right.$\\
        where $c(E)$ & $=
          \left\{
            \begin{array}{l@{\ }l}
              \neg_e(E') & \text{ if } E=\overline{E'},\\
              E & \text{ otherwise.}\\ 
            \end{array}
          \right.$\\
      \end{tabular}
    }
    
    Let us now give the definition of the support operators. For any $k$-ary boolean function $\mathrm{f}$, let $\mathrm{f}_{\mathbb{C}}$ be the operator from $(\mathbb{C})^k$ to $\mathbb{C}$ associated with $\mathrm{f}$  defined by:
    
    \begin{itemize}
      \item $\mathrm{f}_{\mathbb{C}}(\mathcal{C}_1,\ldots,\mathcal{C}_k)=
        \ocup_{b=(b_1,\ldots,b_k)\in \mathbb{B}^k\mid \mathrm{f}(b)=1}
          \ocap_{1\leq j\leq k} \mathrm{g}(b_j,\mathcal{C}_j)
      $,
      \item where $\mathrm{g}(b_j,\mathcal{C}_j)=\left\{
        \begin{array}{l@{\ }l}
          \mathcal{C}_j & \text{ if }b_j=1,\\
          \oneg \mathcal{C}_j & \text{ otherwise.}\\
        \end{array}
      \right.$
    \end{itemize}
    
    The operator $\cdot_{\mathbb{C}}$ from $\mathbb{C}\times \mathrm{Exp}(\Sigma)$ to $\mathbb{C}$ is defined, for any clause $\mathcal{C}$ in $\mathbb{C}$ and for any expression $F$ in $\mathrm{Exp}(\Sigma)$, by:
    
    \centerline{
      $\mathcal{C}\cdot_{\mathbb{C}} F=\bigcup_{C\in \mathcal{C}} \{\mathrm{h}_{\mathbb{C}}(\{C\})\cdot F\}$.
    }
    
    Finally, we consider the operation  set $\mathcal{O}_{\mathbb{C}}$ defined by 
    
    \centerline{
    $\mathcal{O}_{\mathbb{C}}=
         \{\ocup,\cdot_{\mathbb{C}}\}
        \cup
        \{\mathrm{f}_{\mathbb{C}}\mid \mathrm{f} \text{ is a }k\text{-ary boolean function different from }\vee\}
      .$}
      
     Let us notice that, by definition, the operator $\oneg$ (resp. $\ocap$) is equal to  the operator $\neg_{\mathbb{C}}$ (resp. $\wedge_{\mathbb{C}}$) whereas the operator $\ocup$ is different from $\vee_{\mathbb{C}}$, since:
     
     \centerline{ $\mathcal{C}_1\vee_{\mathbb{C}} \mathcal{C}_2= ((\oneg(\mathcal{C}_1) \ocap \mathcal{C}_2) \ocup (\mathcal{C}_1 \ocap \oneg(\mathcal{C}_2))) \ocup (\mathcal{C}_1 \ocap \mathcal{C}_2)$.}
     
     There exist several expressions (combinations of $\ocup$, $\ocap$ and $\oneg$) for a given $\mathrm{f}_{\mathbb{C}}$ operator. As an example, with the ternary $\vee^3(b_1,b_2,b_3)=b_1\vee b_2 \vee b_3$ is associated an operator $\vee^3_{\mathbb{C}}$ that can be expressed as the combination of two $\ocup$ operators:

    \centerline{
      $\vee^3_{\mathbb{C}}(\mathcal{C}_1,\mathcal{C}_2,\mathcal{C}_3)=(\mathcal{C}_1 \ocup \mathcal{C}_2)\ocup \mathcal{C}_3$.
    }
     
     Reduced expressions can be found using Karnaugh maps for instance.
    
     We now consider the $6$-tuple $\mathbb{S}_{\mathbb{C}}=(\Sigma,\mathbb{C}, \mathcal{O}_{\mathbb{C}},\mathrm{h}_{\mathbb{C}},\{\{1\}\},\emptyset)$ and show that it is a support.
    
    \begin{proposition}\label{prop sc support}
      The $6$-tuple $\mathbb{S}_\mathbb{C}$ is a support.
    \end{proposition} 
    \begin{proof}
    
      Properties of support are trivially checked for the clauses $\{\{1\}\}$, $\emptyset$ and $\{\emptyset\}$. Let us consider that $\mathcal{C},\mathcal{C}_1,\ldots,\mathcal{C}_k$ are elements in $\mathbb{C}\setminus\{\emptyset,\{\emptyset\}\}$. According to the definitions of $\ocup$ and $\ocap$:
      
      \centerline{
        \begin{tabular}{l@{\ }l} 
          $L(\mathrm{h}_{\mathbb{C}}(\mathcal{C}_1\ocup \mathcal{C}_2))$ & $=L(\mathrm{h}_{\mathbb{C}}(\mathcal{C}_1\cup \mathcal{C}_2))$\\
          & $=L(\sum_{C\in \mathcal{C}_1\cup \mathcal{C}_2} \bigwedgee_{E\in C} \mathrm{c}(E))$\\
          & $=\bigcup_{C\in \mathcal{C}_1\cup \mathcal{C}_2} L( \bigwedgee_{E\in C} \mathrm{c}(E))$\\
          & $=\bigcup_{C\in \mathcal{C}_1} L( \bigwedgee_{E\in C} \mathrm{c}(E)) \cup \bigcup_{C\in \mathcal{C}_2} L( \bigwedgee_{E\in C} \mathrm{c}(E))$\\
          & $= L( \sum_{C\in \mathcal{C}_1} \bigwedgee_{E\in C} \mathrm{c}(E)) \cup  L( \sum_{C\in \mathcal{C}_2}\bigwedgee_{E\in C} \mathrm{c}(E))$\\
          & $= L(\mathrm{h}_{\mathbb{C}}(\mathcal{C}_1)) \cup  L(\mathrm{h}_{\mathbb{C}}(\mathcal{C}_2))$\\
          & $= L(\mathrm{h}_{\mathbb{C}}(\mathcal{C}_1) + \mathrm{h}_{\mathbb{C}}(\mathcal{C}_2))$\\
        \end{tabular}
      }
      
      \centerline{
        \begin{tabular}{l@{\ }l} 
          $L(\mathrm{h}_{\mathbb{C}}(\mathcal{C}_1\ocap \mathcal{C}_2))$ & $=L(\mathrm{h}_{\mathbb{C}}(\bigcup_{(C_1,C_2)\in \mathcal{C}_1\times\mathcal{C}_2}\{C_1\cup C_2\})$\\
          & $=\bigcup_{(C_1,C_2)\in \mathcal{C}_1\times\mathcal{C}_2} L(\mathrm{h}_{\mathbb{C}}(\{C_1\cup C_2\}))$\\
          & $=\bigcup_{(C_1,C_2)\in \mathcal{C}_1\times\mathcal{C}_2} L(\wedge_e{_{C\in C_1\cup C_2 }} c(E))$\\
          & $=\bigcup_{(C_1,C_2)\in \mathcal{C}_1\times\mathcal{C}_2} L(\wedge_e{_{C\in C_1 }} c(E)) \cap L(\wedge_e{_{C\in C_2 }} c(E))$\\
          & $=\bigcup_{(C_1,C_2)\in \mathcal{C}_1\times\mathcal{C}_2} L(\mathrm{h}_{\mathbb{C}}(\{C_1\})) \cap L(\mathrm{h}_{\mathbb{C}}(\{C_2\}))$\\
          & $=\bigcup_{C\in \mathcal{C}_1} L(\mathrm{h}_{\mathbb{C}}(\{C\})) \cap \bigcup_{C\in \mathcal{C}_2} L(\mathrm{h}_{\mathbb{C}}(\{C\}))$\\
          & $= L(\sum_{C\in \mathcal{C}_1} \mathrm{h}_{\mathbb{C}}(\{C\})) \cap  L(\sum_{C\in \mathcal{C}_2} \mathrm{h}_{\mathbb{C}}(\{C\}))$\\
          & $=L(\mathrm{h}_{\mathbb{C}}(\mathcal{C}_1) \wedge_e \mathrm{h}_{\mathbb{C}}(\mathcal{C}_2))$\\
        \end{tabular}
      }     
      
      Moreover,
      
      \centerline{
        \begin{tabular}{l@{\ }l}
          $L(\mathrm{h}_{\mathbb{C}}(\oneg(\{\{c\}\})))$ & $=L(\mathrm{h}_{\mathbb{C}}(\{\{n(c)\}\}))$\\
          & $=
              \left\{
                \begin{array}{l@{\ }l}
                  L(\mathrm{h}_{\mathbb{C}}(\{\{\overline{E}\}\})) & \text{ if } c=E\\
                  L(\mathrm{h}_{\mathbb{C}}(\{\{E\}\})) & \text{ if } c=\overline{E}\\
                \end{array}
              \right.$\\
          & $=
              \left\{
                \begin{array}{l@{\ }l}
                  L(\neg_e E) & \text{ if } c=E\\
                  L(E) & \text{ if } c=\overline{E}\\
                \end{array}
              \right.$\\
          & $=
              \left\{
                \begin{array}{l@{\ }l}
                  \neg_L(E) & \text{ if } c=E\\
                  \neg_L(L(\neg_e E)) & \text{ if } c=\overline{E}\\
                \end{array}
              \right.$\\
          & $=
              \left\{
                \begin{array}{l@{\ }l}
                  \neg_L(\mathrm{h}_{\mathbb{C}}(\{\{c\}\})) & \text{ if } c=E\\
                  \neg_L(L(\mathrm{h}_{\mathbb{C}}(\{\{c\}\}))) & \text{ if } c=\overline{E}\\
                \end{array}
              \right.$\\
          & $=\neg_L(L(\mathrm{h}_{\mathbb{C}}(\{\{c\}\})))$\\
        \end{tabular}
      }

      Consequently :
      
      \centerline{
        \begin{tabular}{l@{\ }l} 
          $L(\mathrm{h}_{\mathbb{C}}(\oneg(\mathcal{C}_1)))$ & $=L(\mathrm{h}_{\mathbb{C}}(\ocap_{C_1\in\mathcal{C}_1}  \ocup_{c\in C} \{\{\mathrm{n}(c)\} \}))$\\
          & $=\bigcap_{C_1\in\mathcal{C}_1} L(\mathrm{h}_{\mathbb{C}}( \ocup_{c\in C} \{\{\mathrm{n}(c)\} \}))$\\
          & $=\bigcap_{C_1\in\mathcal{C}_1} \bigcup_{c\in C} L(\mathrm{h}_{\mathbb{C}}(\{  \{\mathrm{n}(c)\} \}))$\\
          & $=\bigcap_{C_1\in\mathcal{C}_1} \bigcup_{c\in C} L(\mathrm{h}_{\mathbb{C}}(\oneg(\{\{c\}\})))$\\
          & $=\bigcap_{C_1\in\mathcal{C}_1} \bigcup_{c\in C} \neg_L(L(\mathrm{h}_{\mathbb{C}}(\{\{c\}\})))$\\
          & $=\bigcap_{C_1\in\mathcal{C}_1} \neg_L(\bigcap_{c\in C} L(\mathrm{h}_{\mathbb{C}}(\{\{c\}\})))$\\
          & $=\neg_L(\bigcup_{C_1\in\mathcal{C}_1} \bigcap_{c\in C} L(\mathrm{h}_{\mathbb{C}}(\{\{c\}\})))$\\
          & $=\neg_L(L(\mathrm{h}_{\mathbb{C}}(\mathcal{C}_1)))$\\
          & $=L(\neg_e(\mathrm{h}_{\mathbb{C}}(\mathcal{C}_1)))$\\
        \end{tabular}
      }
      
      Hence, according to definition of $\mathrm{f}_{\mathbb{C}}$:
      
      \centerline{
        \begin{tabular}{l@{\ }l} 
          $L(\mathrm{h}_{\mathbb{C}}(\mathrm{f}_{\mathbb{C}}(\mathcal{C}_1,\ldots,\mathcal{C}_k)))$ & $=L(\mathrm{h}_{\mathbb{C}}(\ocup_{b=(b_1,\ldots,b_k)\mid \mathrm{f}(b)=1}
          \ocap_{1\leq j\leq k} \mathrm{g}(b_j,\mathcal{C}_j)))$\\
          & $=\bigcup_{b=(b_1,\ldots,b_k)\mid \mathrm{f}(b)=1} \bigcap_{1\leq j\leq k} L(\mathrm{h}_{\mathbb{C}}(
           \mathrm{g}(b_j,\mathcal{C}_j)))$\\
          & $= \bigcup_{b=(b_1,\ldots,b_k)\mid \mathrm{f}(b)=1} \bigcap_{1\leq j\leq k}
            \left\{
              \begin{array}{l@{\ }l}
                 L(\mathrm{h}_{\mathbb{C}}(\mathcal{C}_j)) & \text{ if } b_j=1 \\
                L(\mathrm{h}_{\mathbb{C}}(\oneg \mathcal{C}_j)) & \text{ if } b_j=0 \\
              \end{array}
            \right.$\\
          & $= \bigcup_{b=(b_1,\ldots,b_k)\mid \mathrm{f}(b)=1} \bigcap_{1\leq j\leq k}
            \left\{
              \begin{array}{l@{\ }l}
                 L(\mathrm{h}_{\mathbb{C}}(\mathcal{C}_j)) & \text{ if } b_j=1 \\
                \neg_L L(\mathrm{h}_{\mathbb{C}}(\mathcal{C}_j)) & \text{ if } b_j=0 \\
              \end{array}
            \right.$\\
          & $=\mathrm{f}_{L}(L(\mathrm{h}_{\mathbb{C}}(\mathcal{C}_1)),\ldots,L(\mathrm{h}_{\mathbb{C}}(\mathcal{C}_k)))$\\
          & $=L(\mathrm{f}_{e}(\mathrm{h}_{\mathbb{C}}(\mathcal{C}_1),\ldots,\mathrm{h}_{\mathbb{C}}(\mathcal{C}_k)))$\\
        \end{tabular}
      }
      
      Furthermore,
      
      \centerline{
        \begin{tabular}{l@{\ }l}
          $L(\mathrm{h}_{\mathbb{C}}(\mathcal{C}\cdot_{\mathbb{C}} F))$ & $=L(\mathrm{h}_{\mathbb{C}}(\bigcup_{C\in \mathcal{C}} \{(\bigwedgee_{E\in C} c(E) )\cdot F\}))$\\
          & $=\bigcup_{C\in \mathcal{C}} L( (\bigwedgee_{E\in C} c(E) )\cdot F)$\\
          & $=\bigcup_{C\in \mathcal{C}} L( \bigwedgee_{E\in C} c(E) ) \cdot L(F)$\\
          & $=(\bigcup_{C\in \mathcal{C}} L( \bigwedgee_{E\in C} c(E) )) \cdot L(F)$\\
          & $= L( \sum_{C\in \mathcal{C}} \bigwedgee_{E\in C} c(E) ) \cdot L(F)$\\
          & $= L(\mathrm{h}_{\mathbb{C}}(\mathcal{C}) ) \cdot L(F)$\\
          & $= L(\mathrm{h}_{\mathbb{C}}(\mathcal{C}) \cdot F)$\\
        \end{tabular}
      }

    \end{proof}

     We now study the properties of the derivation $\mathrm{D}_\mathbb{C}$ associated with the support $\mathbb{S}_\mathbb{C}$.

    \begin{theorem}\label{thm db afa le}
      Let $E$ be an expression in $\mathrm{Exp}(\Sigma)$. Let $(\mathrm{D},\mathrm{B})$ be a couple in $\{\mathrm{D}_{\mathbb{C}}\}\times\{\mathrm{B}_A,\mathrm{B}_B,\mathrm{B}_C\}$. Then:
      
      \centerline{
        The $(\mathrm{D},\mathrm{B})$-automaton of $E$ is an AFA recognizing $L(E)$.
      }
    \end{theorem}
    \begin{proof}
    
      \textbf{(I)} Let us show that the derivation $\mathrm{D}_{\mathbb{C}}$ satisfies the sufficient conditions for finiteness of Proposition~\ref{prop finitude}. \textbf{(a)} Since $\ocup=\cup$, the function $\ocup$ is associative, commutative and idempotent (\textbf{H1}). \textbf{(b)} According to the definition of the operators $\ocup$, $\ocap$ and $\oneg$:
      
        \centerline{
          \begin{tabular}{l@{\ }l}
            $\mathrm{D}_{\mathbb{C}}(a,\mathrm{h}_{\mathbb{C}}(\mathcal{C}_1\ocup \mathcal{C}_2))$ & $=\mathrm{D}_{\mathbb{C}}(a,\mathrm{h}_{\mathbb{C}}(\mathcal{C}_1\cup \mathcal{C}_2))$ \\
            & $=\mathrm{D}_{\mathbb{C}}(a,\sum_{C\in \mathcal{C}_1\cup \mathcal{C}_2}
          \left\{
            \begin{array}{l@{\ }l}
              \neg_e 0 & \text{ if }C=\emptyset,\\
              \bigwedgee_{E\in C} c(E) & \text{ otherwise,}
            \end{array}
          \right.)$ \\
          & $=\bigcup_{C\in \mathcal{C}_1\cup \mathcal{C}_2} 
          \left\{
            \begin{array}{l@{\ }l}
              \mathrm{D}_{\mathbb{C}}(a,\neg_e 0) & \text{ if }C=\emptyset,\\
              \mathrm{D}_{\mathbb{C}}(a,\bigwedgee_{E\in C} c(E)) & \text{ otherwise,}
            \end{array}
          \right.$ \\
          & $=\bigcup_{C\in \mathcal{C}_1} 
          \left\{
            \begin{array}{l@{\ }l}
              \mathrm{D}_{\mathbb{C}}(a,\neg_e 0) & \text{ if }C=\emptyset,\\
              \mathrm{D}_{\mathbb{C}}(a,\bigwedgee_{E\in C} c(E)) & \text{ otherwise,}
            \end{array}
          \right.$ \\
          & $\ \ \ \ \cup \bigcup_{C\in \mathcal{C}_2} 
          \left\{
            \begin{array}{l@{\ }l}
              \mathrm{D}_{\mathbb{C}}(a,\neg_e 0) & \text{ if }C=\emptyset,\\
              \mathrm{D}_{\mathbb{C}}(a,\bigwedgee_{E\in C} c(E)) & \text{ otherwise,}
            \end{array}
          \right.$\\
          & $=\mathrm{D}_{\mathbb{C}}(a, \sum_{C\in \mathcal{C}_1} 
          \left\{
            \begin{array}{l@{\ }l}
              \neg_e 0 & \text{ if }C=\emptyset,\\
              \bigwedgee_{E\in C} c(E) & \text{ otherwise,}
            \end{array}
          \right.$ \\
          & $\ \ \ \ \cup \mathrm{D}_{\mathbb{C}}(a, \sum_{C\in \mathcal{C}_2} 
          \left\{
            \begin{array}{l@{\ }l}
              \neg_e 0 & \text{ if }C=\emptyset,\\
              \bigwedgee_{E\in C} c(E) & \text{ otherwise,}
            \end{array}
          \right.$\\
            & $=\mathrm{D}_\mathbb{C}(a,\mathrm{h}_{\mathbb{C}}(\mathcal{C}_1))\ocup \mathrm{D}_{\mathbb{C}}(a,\mathrm{h}_{\mathbb{C}}(\mathcal{C}_2))$.\\
          \end{tabular}
        }
      
        \centerline{
          \begin{tabular}{l@{\ }l}
            $\mathrm{D}_\mathbb{C}(a,\mathrm{h}_{\mathbb{C}}(\mathcal{C}_1\ocap \mathcal{C}_2))$ & $=\mathrm{D}_\mathbb{C}(a,\mathrm{h}_{\mathbb{C}}(\bigcup_{(C_1,C_2)\in \mathcal{C}_1\times\mathcal{C}_2}\{C_1\cup C_2\})$\\
            & $=\bigcup_{(C_1,C_2)\in \mathcal{C}_1\times\mathcal{C}_2} \mathrm{D}_\mathbb{C}(a,\mathrm{h}_{\mathbb{C}}(\{C_1\cup C_2\})$\\
            & $=\bigcup_{(C_1,C_2)\in \mathcal{C}_1\times\mathcal{C}_2} \mathrm{D}_\mathbb{C}(a,\mathrm{h}_{\mathbb{C}}(\{C_1\}\cup \{C_2\})$\\
            & $=\bigcup_{(C_1,C_2)\in \mathcal{C}_1\times\mathcal{C}_2} \mathrm{D}_\mathbb{C}(a,\mathrm{h}_{\mathbb{C}}(\{C_1\})\cup \mathrm{D}_\mathbb{C}(a,\mathrm{h}_{\mathbb{C}}(\{C_2\})$\\ 
            & $=\mathrm{D}_\mathbb{C}(a,\mathrm{h}_{\mathbb{C}}(\mathcal{C}_1))\ocap \mathrm{D}_\mathbb{C}(a,\mathrm{h}_{\mathbb{C}}(\mathcal{C}_2))$.\\
          \end{tabular}
        }
        
        \centerline{
          \begin{tabular}{l@{\ }l}
            $\mathrm{D}_\mathbb{C}(a,\mathrm{h}_{\mathbb{C}}(\oneg(\mathcal{C}_1)))$ & $=\mathrm{D}_\mathbb{C}(a,\mathrm{h}_{\mathbb{C}}(\ocap_{C_1\in\mathcal{C}_1}  \ocup_{c\in C} \{\{\mathrm{n}(c)\} \}))$\\
            & $=\ocap_{C_1\in\mathcal{C}_1} \mathrm{D}_\mathbb{C}(a,\mathrm{h}_{\mathbb{C}}(\{ \ocup_{c\in C} \{\mathrm{n}(c)\} \}))$\\
            & $=\ocap_{C_1\in\mathcal{C}_1} \ocup_{c\in C} \mathrm{D}_\mathbb{C}(a,\mathrm{h}_{\mathbb{C}}(\{  \{\mathrm{n}(c)\} \}))$\\
            & $=\ocap_{C_1\in\mathcal{C}_1} \ocup_{c\in C} \mathrm{D}_\mathbb{C}(a,\mathrm{h}_{\mathbb{C}}(
            \left\{
            \begin{array}{l@{\ }l}
              c' & \text{ if }c=\overline{c'},\\
              \overline{c} & \text{ otherwise.}\\
            \end{array}\right.
          ))
            $\\
            & $=\ocap_{C_1\in\mathcal{C}_1} \ocup_{c\in C} 
            \left\{
            \begin{array}{l@{\ }l}
              \mathrm{D}_\mathbb{C}(a,\mathrm{h}_{\mathbb{C}}(c')) & \text{ if }c=\overline{c'},\\
              \mathrm{D}_\mathbb{C}(a,\mathrm{h}_{\mathbb{C}}(\overline{c})) & \text{ otherwise.}\\
            \end{array}\right.
            $\\            
            & $=\ocap_{C_1\in\mathcal{C}_1} \ocup_{c\in C} 
            \left\{
            \begin{array}{l@{\ }l}
              \mathrm{D}_\mathbb{C}(a,c') & \text{ if }c=\overline{c'},\\
              \mathrm{D}_\mathbb{C}(a,\neg_e c)) & \text{ otherwise.}\\
            \end{array}\right.
            $\\           
            & $=\ocap_{C_1\in\mathcal{C}_1} \ocup_{c\in C} 
            \left\{
            \begin{array}{l@{\ }l}
              \oneg \oneg \mathrm{D}_\mathbb{C}(a,c') & \text{ if }c=\overline{c'},\\
              \oneg \mathrm{D}_\mathbb{C}(a,c)) & \text{ otherwise.}\\
            \end{array}\right.
            $\\          
            & $=\ocap_{C_1\in\mathcal{C}_1} \ocup_{c\in C} 
            \left\{
            \begin{array}{l@{\ }l}
              \oneg \mathrm{D}_\mathbb{C}(a,\neg_e c') & \text{ if }c=\overline{c'},\\
              \oneg \mathrm{D}_\mathbb{C}(a,c)) & \text{ otherwise.}\\
            \end{array}\right.
            $\\
            & $=\ocap_{C_1\in\mathcal{C}_1} \ocup_{c\in C} \oneg \mathrm{D}_\mathbb{C}(a,\mathrm{h}_{\mathbb{C}}(\{  \{c\} \}))$\\
            & $=\oneg \ocup_{C_1\in\mathcal{C}_1} \ocap_{c\in C}  \mathrm{D}_\mathbb{C}(a,\mathrm{h}_{\mathbb{C}}(\{  \{c\} \}))$\\
            & $=\oneg  \mathrm{D}_\mathbb{C}(a,\sum_{C_1\in\mathcal{C}_1} \wedge_e{_{c\in C}}\mathrm{h}_{\mathbb{C}}(\{  \{c\} \}))$\\
            & $=\oneg(\mathrm{D}_\mathbb{C}(a,\mathrm{h}_{\mathbb{C}}(\mathcal{C}_1)))$.\\
          \end{tabular}
        }
        
        Finally, since any $\mathrm{f}_{\mathbb{B}}$ is defined as combination of $\ocup$, $\ocap$ and $\oneg$ operators, hypothesis \textbf{H2} holds. According to Proposition~\ref{prop finitude}, $\mathrm{D}_\mathbb{C}$ is finite.
      
      \textbf{(II)} Let us show that any couple in $\{\mathrm{D}_\mathbb{C}\}\times\{\mathrm{B}_A,\mathrm{B}_B,\mathrm{B}_C\}$ satisfies the atom-derivability property. \textbf{(a)} By definition of $\mathrm{B}_B$, the atom-derivability property is satisfied by $(\mathrm{D}_\mathbb{C},\mathrm{B}_B)$. \textbf{(b)} By induction over the structure of $E$. Let $a$ be a symbol in $\Sigma$. \textbf{(i)} If $E\in \Sigma \cup\{1,0\}$ or if $E=F\cdot G$ or if $E=F^*$, since $\mathrm{B}_A(E)=\mathrm{B}_C(E)=E$, $\mathrm{Atom}(\mathrm{B}_A(E))=\mathrm{Atom}(\mathrm{B}_C(E))=\{E\}$. \textbf{(ii)} Let $\mathrm{f}$ be a $k$-ary boolean function.
       Let us first prove that the equation of the atom-derivanility property is satisfied for the operators $\ocup$, $\ocap$ and $\oneg$. If $\mathcal{C}_1$ or $\mathcal{C}_2$ equal $\emptyset$ or $\{\emptyset\}$, equation is trivially satisfied. Let $\mathcal{C}_1$ and $\mathcal{C}_2$ be two clauses different from $\emptyset$ and $\{\emptyset\}$.
       
      \centerline{
        \begin{tabular}{l@{\ }l}
        $\mathrm{Atom}( \mathrm{B}_{C} (\mathrm{h}_\mathbb{C} (\mathcal{C}_1 \cup \mathcal{C}_2)))$ & $=\mathrm{Atom}( \mathrm{B}_{C} (\sum_{C\in \mathcal{C}_1 \cup \mathcal{C}_2} \wedge_e{_{E\in C}} c(E) ))$\\
        & $=\mathrm{Atom}( \vee_\mathbb{B}{_{C\in \mathcal{C}_1 \cup \mathcal{C}_2}} \mathrm{B}_{C} ( \wedge_e{_{E\in C}} c(E) ))$\\
        & $=\bigcup{_{C\in \mathcal{C}_1 \cup \mathcal{C}_2}} \mathrm{Atom}(  \mathrm{B}_{C} ( \wedge_e{_{E\in C}} c(E) ))$\\
        & $=\bigcup{_{C\in \mathcal{C}_1}} \mathrm{Atom}(  \mathrm{B}_{C} ( \wedge_e{_{E\in C}} c(E) ))$\\
        & $\ \ \ \cup \bigcup{_{C\in \mathcal{C}_2}} \mathrm{Atom}(  \mathrm{B}_{C} ( \wedge_e{_{E\in C}} c(E) ))$\\
        & $ =\mathrm{Atom}( \vee_\mathbb{B}{_{C\in \mathcal{C}_1}} \mathrm{B}_{C} ( \wedge_e{_{E\in C}} c(E) ))$\\
        & $\ \ \ \cup  \mathrm{Atom}(\vee_\mathbb{B}{_{C\in \mathcal{C}_2}}  \mathrm{B}_{C} ( \wedge_e{_{E\in C}} c(E) ))$\\
        & $= \mathrm{Atom}(  \mathrm{B}_{C} (\sum{_{C\in \mathcal{C}_1}} \wedge_e{_{E\in C}} c(E) ))$\\
        & $\ \ \ \cup  \mathrm{Atom}(  \mathrm{B}_{C} ( \sum{_{C\in \mathcal{C}_2}} \wedge_e{_{E\in C}} c(E) ))$\\
        & $=\mathrm{Atom}( \mathrm{B}_{C} (\mathrm{h}_\mathbb{C} (\mathcal{C}_1))) \cup \mathrm{Atom}( \mathrm{B}_{C} (\mathrm{h}_\mathbb{C} (\mathcal{C}_2)))$\\
        \end{tabular}
      }
      
      \centerline{
        \begin{tabular}{l@{\ }l}
        $\mathrm{Atom}( \mathrm{B}_{C} (\mathrm{h}_\mathbb{C} (\mathcal{C}_1 \ocap \mathcal{C}_2)))$ & $=\mathrm{Atom}( \mathrm{B}_{C} (\mathrm{h}_\mathbb{C} (\bigcup_{(C_1,C_2)\in\mathcal{C}_1\times\mathcal{C}_2}\{C_1\cup C_2\})))$\\
        & $=\bigcup_{(C_1,C_2)\in\mathcal{C}_1\times\mathcal{C}_2} \mathrm{Atom}( \mathrm{B}_{C} (\mathrm{h}_\mathbb{C} (\{C_1\cup C_2\})))$\\
        & $=\bigcup_{(C_1,C_2)\in\mathcal{C}_1\times\mathcal{C}_2} \mathrm{Atom}( \mathrm{B}_{C} (\mathrm{h}_\mathbb{C} (\{C_1\})))$\\
        & $\ \ \ \cup \mathrm{Atom}( \mathrm{B}_{C} (\mathrm{h}_\mathbb{C} (\{C_2\})))$\\        
        & $=\bigcup_{(C_1,C_2)\in\mathcal{C}_1\times\mathcal{C}_2} \mathrm{Atom}( \mathrm{B}_{C} (\mathrm{h}_\mathbb{C} (\{C_1\})))$\\
        & $\ \ \ \cup \bigcup_{(C_1,C_2)\in\mathcal{C}_1\times\mathcal{C}_2} \mathrm{Atom}( \mathrm{B}_{C} (\mathrm{h}_\mathbb{C} (\{C_2\}))) $\\
        & $=\bigcup_{C_1\in\mathcal{C}_1} \mathrm{Atom}( \mathrm{B}_{C} (\mathrm{h}_\mathbb{C} (\{C_1\})))$\\
        & $\ \ \ \cup \bigcup_{C_2\in\mathcal{C}_2} \mathrm{Atom}( \mathrm{B}_{C} (\mathrm{h}_\mathbb{C} (\{C_2\}))) $\\
        & $=\mathrm{Atom}( \mathrm{B}_{C} (\mathrm{h}_\mathbb{C} (\mathcal{C}_1))) \cup \mathrm{Atom}( \mathrm{B}_{C} (\mathrm{h}_\mathbb{C} (\mathcal{C}_2)))$\\
        \end{tabular}
      }

      \centerline{
        \begin{tabular}{l@{\ }l}
        $\mathrm{Atom}( \mathrm{B}_{C} (\mathrm{h}_\mathbb{C} (\oneg \mathcal{C}_1)))$ & $=\mathrm{Atom}( \mathrm{B}_{C} (\mathrm{h}_\mathbb{C} (\{\{n(x)\}\})))$\\
        & $=\bigcup_{C\in\mathcal{C}_1}\bigcup_{x\in C}\mathrm{Atom}( \mathrm{B}_{C} (\mathrm{h}_\mathbb{C} ( \{\{n(x)\}\})))$\\
        & $=\bigcup_{C\in\mathcal{C}_1}\bigcup_{x\in C}\mathrm{Atom}( \mathrm{B}_{C} (\mathrm{h}_\mathbb{C} ( \{\{x\}\})))$\\
        & $=\mathrm{Atom}( \mathrm{B}_{C} (\mathrm{h}_\mathbb{C} ( \bigcup_{C\in\mathcal{C}_1}\bigcup_{x\in C} \{\{x\}\})))$\\
        & $=\mathrm{Atom}( \mathrm{B}_{C} (\mathrm{h}_\mathbb{C} ( \mathcal{C}_1)))$\\
        \end{tabular}
      }
      
      Hence, since any operator $\mathrm{f}_\mathbb{C}$ is a composition of $\ocup$, $\ocap$ and $\oneg$:
      
      \centerline{
        $\mathrm{Atom}( \mathrm{B}_{C} (\mathrm{h}_\mathbb{C} (\mathrm{f}_\mathbb{C}(\mathcal{C}_1,\ldots,\mathcal{C}_k))) )
        =
        \bigcup_{1\leq j\leq k} \mathrm{Atom}( \mathrm{B}_{C} (\mathrm{h}_\mathbb{C} (\mathcal{C}_j)))
        $.
      }

      Consequently:
      
      \centerline{
        \begin{tabular}{l@{\ }l}
          & $\mathrm{Atom}(\mathrm{B}_C(\mathrm{h}_{\mathbb{C}}(\mathrm{D}_\mathbb{C}(a,\mathrm{f}_e(E_1,\ldots,E_k)))))$\\ & $=\mathrm{Atom}(\mathrm{B}_C(\mathrm{h}_{\mathbb{C}}(\mathrm{f}_\mathbb{C}(\mathrm{D}_\mathbb{C}(a,E_1),\ldots,\mathrm{D}_\mathbb{C}(a,E_k)))))$\\
          & $=\bigcup_{1\leq j\leq k}\mathrm{Atom}(\mathrm{B}_C(\mathrm{h}_{\mathbb{C}}(\mathrm{D}_\mathbb{C}(a,E_j))))$\\
          & $=\bigcup_{1\leq j\leq k} \bigcup_{E'\in \mathrm{Atom}(\mathrm{B}_{\mathrm{C}}(E_j))} \mathrm{Atom}(\mathrm{B}_C(\mathrm{h}_{\mathbb{C}}(\mathrm{D}_\mathbb{C}(a,E'))))$\\
          & $=\bigcup_{E'\in \mathrm{Atom}(\mathrm{B}_{\mathrm{C}}(E))} \mathrm{Atom}(\mathrm{B}_C(\mathrm{h}_{\mathbb{C}}(\mathrm{D}_\mathbb{C}(a,E'))))$\\
        \end{tabular}
      }
      
      Furthermore, if $\mathrm{f}\neq \vee$, $\mathrm{B}_A(\mathrm{f}_e(E_1,\ldots,E_k))=\mathrm{f}_e(E_1,\ldots,E_k)$.
      
      \centerline{
        $\mathrm{Atom}(\mathrm{B}_A(\mathrm{f}_e(E_1,\ldots,E_k)))=\{\mathrm{f}_e(E_1,\ldots,E_k)\}$.
      }
      
      Finally,
      
      \centerline{
        \begin{tabular}{l@{\ }l}
        $\mathrm{Atom}( \mathrm{B}_{A} (\mathrm{h}_\mathbb{C} (\mathcal{C}_1 \cup \mathcal{C}_2)))$ & $=\mathrm{Atom}( \mathrm{B}_{A} (\sum_{C\in \mathcal{C}_1 \cup \mathcal{C}_2} \wedge_e{_{E\in C}} c(E) ))$\\
        & $=\mathrm{Atom}( \vee_\mathbb{B}{_{C\in \mathcal{C}_1 \cup \mathcal{C}_2}} \mathrm{B}_{A} ( \wedge_e{_{E\in C}} c(E) ))$\\
        & $=\bigcup{_{C\in \mathcal{C}_1 \cup \mathcal{C}_2}} \mathrm{Atom}(  \mathrm{B}_{A} ( \wedge_e{_{E\in C}} c(E) ))$\\
        & $=\bigcup{_{C\in \mathcal{C}_1}} \mathrm{Atom}(  \mathrm{B}_{A} ( \wedge_e{_{E\in C}} c(E) ))$\\
        & $\ \ \ \cup \bigcup{_{C\in \mathcal{C}_2}} \mathrm{Atom}(  \mathrm{B}_{A} ( \wedge_e{_{E\in C}} c(E) ))$\\
        & $ =\mathrm{Atom}( \vee_\mathbb{B}{_{C\in \mathcal{C}_1}} \mathrm{B}_{A} ( \wedge_e{_{E\in C}} c(E) ))$\\
        & $\ \ \ \cup  \mathrm{Atom}(\vee_\mathbb{B}{_{C\in \mathcal{C}_2}}  \mathrm{B}_{A} ( \wedge_e{_{E\in C}} c(E) ))$\\
        & $= \mathrm{Atom}(  \mathrm{B}_{A} (\sum{_{C\in \mathcal{C}_1}} \wedge_e{_{E\in C}} c(E) ))$\\
        & $\ \ \ \cup  \mathrm{Atom}(  \mathrm{B}_{A} ( \sum{_{C\in \mathcal{C}_2}} \wedge_e{_{E\in C}} c(E) ))$\\
        & $=\mathrm{Atom}( \mathrm{B}_{A} (\mathrm{h}_\mathbb{C} (\mathcal{C}_1))) \cup \mathrm{Atom}( \mathrm{B}_{A} (\mathrm{h}_\mathbb{C} (\mathcal{C}_2)))$\\
        \end{tabular}
      }
      
      Consequently, any couple in $\{\mathrm{D}_\mathbb{C}\}\times\{\mathrm{B}_A,\mathrm{B}_B,\mathrm{B}_C\}$ satisfies the atom-derivability property.
      
      \textbf{(III)} According to Theorem~\ref{thm afa}, from \textbf{(I)} and \textbf{(II)}, the theorem holds.      
    \end{proof}

    The following example illustrates the computation of an AFA from a regular expression. In order to improve readability, regular expressions are simplified according to the following rules:
    
    \centerline{
      \begin{tabular}{l}
        $E+0 \equiv 0 +E\equiv E$\\
        $E\cdot 0 \equiv 0 \cdot E\equiv 0$\\
        $E\cdot 1 \equiv 1 \cdot E\equiv E$\\
      \end{tabular}
    }
    
    \begin{example}
      Let $E=((ab)^*a) \mathrm{XOR}_e ((abab)^*a)$. We now construct the  $(\mathrm{D}_\mathbb{C},\mathrm{B}_C)$-automaton of $E$, that is the AFA $(\Sigma,Q,E,F,\delta)$. We first compute the derivatives $\mathcal{C}$ of $E$ w.r.t. any symbol in $\Sigma$, according to the derivation $\mathrm{D}_\mathbb{C}$ and then compute the derivatives of the expressions that are atoms of the base of $\mathrm{h}_{\mathbb{C}}(\mathcal{C})$. This scheme is repeated until no more expression is produced.

      \centerline{\begin{tabular}{l@{\ }l}
        $\mathrm{D}_\mathbb{C}(a,E)$ & $=(\oneg(\mathrm{D}_\mathbb{C}(a,(ab)^*a)) \ocap \mathrm{D}_\mathbb{C}(a,(abab)^*a))$\\
        & $\ \  \ocup (\mathrm{D}_\mathbb{C}(a,(ab)^*a) \ocap \oneg(\mathrm{D}_\mathbb{C}(a,(abab)^*a))) $\\
        & $=(\oneg(\{\{b(ab)^*a\},\{1\}\}) \ocap \{\{bab(abab)^*a\},\{1\}\})$\\
        & $\ \  \ocup (\{\{b(ab)^*a\},\{1\}\} \ocap \oneg(\{\{bab(abab)^*a\},\{1\}\})) $\\
        & $=(\{\{\overline{b(ab)^*a},\overline{1}\}\} \ocap \{\{bab(abab)^*a\},\{1\}\})$\\
        & $\ \  \ocup (\{\{b(ab)^*a\},\{1\}\} \ocap \{\{\overline{bab(abab)^*a},\overline{1}\}\}) $\\
        & $=\{\{\overline{b(ab)^*a},\overline{1},bab(abab)^*a\},\{\overline{b(ab)^*a},\overline{1},1\},$\\
        & $\ \ \{b(ab)^*a,\overline{bab(abab)^*a},\overline{1}\},\{1,\overline{bab(abab)^*a},\overline{1}\}\} $\\        
      \end{tabular}}
      
      \begin{minipage}{0.49\linewidth}
      \begin{tabular}{l@{}l}
        $\mathrm{D}_\mathbb{C}(b,E)$ & $=\emptyset$\\
        $\mathrm{D}_\mathbb{C}(a,b(ab)^*a)$ & $=\emptyset$\\
        $\mathrm{D}_\mathbb{C}(b,b(ab)^*a)$ & $=\{\{(ab)^*a\}\}$\\
        $\mathrm{D}_\mathbb{C}(a,(ab)^*a)$ & $=\{\{b(ab)^*a\},$\\
        & $\{1\}\}$\\
        $\mathrm{D}_\mathbb{C}(b,(ab)^*a)$ & $=\emptyset$\\
        $\mathrm{D}_\mathbb{C}(a,1)$ & $=\emptyset$\\
        $\mathrm{D}_\mathbb{C}(b,1)$ & $=\emptyset$\\
        $\mathrm{D}_\mathbb{C}(a,bab(abab)^*a)$ & $=\emptyset$\\
      \end{tabular}
      \end{minipage}
      \hfill      
      \begin{minipage}{0.49\linewidth}
      \begin{tabular}{l@{}l}
        $\mathrm{D}_\mathbb{C}(b,bab(abab)^*a)$ & $=\{\{ab(abab)^*a\}\}$\\
        $\mathrm{D}_\mathbb{C}(a,ab(abab)^*a)$ & $=\{\{b(abab)^*a\}\}$\\
        $\mathrm{D}_\mathbb{C}(b,ab(abab)^*a)$ & $=\emptyset$\\
        $\mathrm{D}_\mathbb{C}(a,b(abab)^*a)$ & $=\emptyset$\\
        $\mathrm{D}_\mathbb{C}(b,b(abab)^*a)$ & $=\{\{(abab)^*a\}\}$\\
        $\mathrm{D}_\mathbb{C}(a,(abab)^*a)$ & $=\{\{bab(abab)^*a\},$\\
        &$ \{1\}\}$\\
        $\mathrm{D}_\mathbb{C}(b,(abab)^*a)$ & $=\emptyset$\\
      \end{tabular}
      \end{minipage}
      
      From the computation of the derivatives, we deduce:
      
      \begin{itemize}
      \item the set $Q=\{q_1,\ldots,q_8\}$ of states :
      
      \centerline{
        $q_1 =E$, $q_2 =b(ab)^*a$, $q_3 =1$, $q_4 =bab(abab)^*a$,
      }
      
      \centerline{
        $q_5 =ab(abab)^*a$, $q_6 =b(abab)^*a$, $q_7=(ab)^*a, q_8=(abab)^*a$,
      }
      
      \item the function $F$ from $Q$ to $\mathbb{B}$:
      
      \centerline{
        $F(q_3)=F(q_7)=F(q_8)=1$,
      }
      
      \centerline{
        $F(q_1)=F(q_2)=F(q_4)=F(q_5)=F(q_6)=0$,
      }
      
      \item and the function $\delta$ from $Q \times \Sigma$ to $\mathrm{BoolForm}(Q)$:
      
      \centerline{\begin{tabular}{|c||c|c|c|c|c|c|c|c|}
      \hline
        & $q_1$ & $q_2$ & $q_3$ & $q_4$ & $q_5$ & $q_6$ & $q_7$ & $q_8$\\
      \hline\hline
      $a$ & 
        \begin{tabular}{c}
          $\neg_{\mathbb{B}} q_2 \wedge_{\mathbb{B}}  \neg_{\mathbb{B}} q_3 \wedge_{\mathbb{B}} q_4 $\\
          $\vee_\mathbb{B}$\\
          $\neg_{\mathbb{B}} q_2 \wedge_{\mathbb{B}}  \neg_{\mathbb{B}} q_3 \wedge_{\mathbb{B}} q_3$\\
          $\vee_\mathbb{B}$\\
          $ q_2 \wedge_{\mathbb{B}}  \neg_{\mathbb{B}} q_4 \wedge_{\mathbb{B}} \neg_{\mathbb{B}} q_3$\\
          $\vee_\mathbb{B}$\\
          $ q_3 \wedge_{\mathbb{B}}  \neg_{\mathbb{B}} q_4 \wedge_{\mathbb{B}} \neg_{\mathbb{B}} q_3$\\
        \end{tabular}
         &
         $0$
         &
         $0$
         &
         $0$
         &
         $q_6$
         &
         $0$
         &
         $q_2\vee_{\mathbb{B}} q_3$
         &
         $q_4\vee_{\mathbb{B}} q_3$
         \\
         \hline
      $b$ & $0$ & $q_7$ & $0$ & $q_5$ & $0$ & $q_8$ & $0$ & $0$\\
      \hline
      \end{tabular}}      
      \end{itemize}      
      
      Let us notice that in this example, substituting $\mathrm{B}_\mathbb{C}(E)$ to $E$ in $I$ would  produce a smaller automaton.
    \end{example}

\section{Conclusion}
This paper provides two main results.
First, the theoretical scheme of derivations {\it via} a support 
allows us to formalize intrinsic properties of (unrestricted) regular expression derivations.
As a by-product we obtain a kind of unification of the classical derivations that
compute word derivatives, partial derivatives or extended partial derivatives. 
Secondly, the notion of base function that associates a boolean formula with a regular expression
allows us to show how to deduce an alternating automaton equivalent to a given regular expression
from the set of its derivatives {\it via} a given support.
We are now investigating new features: it is possible, for example,
to define morphisms from one support to another one
in order to study the relations between the associated automata.
An other perspective is to replace the derivation mapping by an other mapping (right derivation
or left-and-right derivation for example, or any transformation with good properties).
There also exist well-know algorithms to reduce boolean formulas (Karnaugh, Quine-McCluskey);
we intend to investigate reduction techniques based on derivation.
Finally we intend to extend the theoretical derivation scheme in order to handle the derivation of regular expression 
with multiplicities.
 
  \bibliography{biblio}

\end{document}